\documentclass[10pt, twocolumn, journal]{IEEEtran}
\usepackage{enumerate}
\usepackage{graphicx}
\usepackage{epsf}
\usepackage{subfigure}
\usepackage{amsmath}
\usepackage{amssymb}
\usepackage{array}
\usepackage{setspace}
\usepackage[amsmath,thmmarks,amsthm]{ntheorem}
\usepackage[ruled,lined]{algorithm2e}
\usepackage{algorithmic}
\usepackage{color}
\usepackage{amsfonts}
\usepackage{dsfont}
\usepackage{xfrac}
\usepackage{breqn}
\usepackage{bbm}

\usepackage{cite}
\graphicspath{{Fig/},{fig/}}

\newcommand{\tabincell}[2]{\begin{tabular}{@{}#1@{}}#2\end{tabular}}

\theoremseparator{.}

\newtheorem{definition}{Definition}

\newtheorem{lemma}{Lemma}

\DeclareMathOperator{\diag}{diag}

\newcommand{\m}{3.3}

\begin{document}

\title{\huge MARS: Message Passing for Antenna and RF Chain Selection for Hybrid Beamforming in MIMO Communication Systems}

\author{
Li-Hsiang Shen,~\IEEEmembership{Member,~IEEE},
Yen-Chun Lo,
Kai-Ten Feng,~\IEEEmembership{Senior Member,~IEEE},
Sau-Hsuan Wu,~\IEEEmembership{Member,~IEEE},
Lie-Liang Yang,~\IEEEmembership{Fellow,~IEEE}}


\maketitle

\begin{abstract}
In this paper, we consider a prospective receiving hybrid beamforming structure consisting of several radio frequency (RF) chains and abundant antenna elements in multi-input multi-output (MIMO) systems. Due to conventional costly full connections, we design an enhanced partially connected beamformer employing a low-density parity-check (LDPC)-based structure. As a benefit of the LDPC-based structure, information can be exchanged among clustered RF/antenna groups, which results in a low computational complexity order. Advanced message passing (MP) capable of inferring and transferring information among different paths is designed to support the LDPC-based hybrid beamformer. We propose a message-passing enhanced antenna and RF chain selection (MARS) scheme for minimizing the operational power of antennas and RF chains of the receiver as well as hybrid beamforming. Furthermore, sequential and parallel MP schemes for MARS are designed, namely, MARS-S and MARS-P, respectively, to address the convergence speed issue. A heuristic genetic algorithm is designed for receiving hybrid beamforming, comprising gene generation initialization, elite selection, crossover, and mutation. Simulations validate the convergence of both the MARS-P and the MARS-S algorithms. Due to the asynchronous information transfer of MARS-P, it requires higher power than MARS-S, which strikes a compelling balance among power consumption, convergence, and computational complexity. It is also demonstrated that the proposed MARS scheme outperforms the existing benchmarks using the heuristic method of fully/partially connected architectures in the open literature by requiring the lowest power and realizing the highest energy efficiency.
\end{abstract}

\begin{IEEEkeywords}
Hybrid beamforming, MIMO, LDPC, power minimization, message passing.
\end{IEEEkeywords}

{\let\thefootnote\relax\footnotetext
{Li-Hsiang Shen is with the Department of Communication Engineering, National Central University, Taoyuan 320317, Taiwan. (email: shen@ncu.edu.tw)}}

{\let\thefootnote\relax\footnotetext
{Yen-Chun Lo, Kai-Ten Feng, and Sau-Hsuan Wu are with the Department of Electronics and Electrical Engineering, National Yang Ming Chiao Tung University (NYCU), Hsinchu 300093, Taiwan. (email: tommylo0915@gmail.com, ktfeng@nycu.edu.tw, and 
sauhsuan@nycu.edu.tw)}}

{\let\thefootnote\relax\footnotetext
{Lie-Liang Yang is with the Department of Electronics and Computer Science, University of Southampton, Southampton SO17 1BJ, U.K. (e-mail: lly@ecs.soton.ac.uk)
}}

\section{Introduction}
Wireless communication systems are experiencing revolutionary advances due to rapid changes in the creation and sharing of diverse information by human society. Abundant applications with exponentially high wireless traffic demands are emerging \cite{1,acm}. In recent years, researchers have investigated promising technology for supporting advanced wireless communication systems with high spectral efficiency and energy efficiency (EE). A denser deployment of base stations (BSs) with a larger bandwidth can achieve high performance with increasing degrees of freedom from the perspective of hardware and radio resources but requires substantial consumption of electricity \cite{2}. Accordingly, serious atmospheric pollution issues will arise due to global carbon emissions ascribed to the information and communications technology industry. Therefore, it would be beneficial if BSs were capable of efficiently employing low-power configurations to meet stringent service requirements and resiliently improve network performance \cite{3}.

	As a prospect of wireless transmission systems, a massive multi-input multi-output (MIMO) system deployed within a BS is capable of transceiving enormous radio frequency (RF) signals, benefiting from its gains from multiplexing and diversity \cite{4}. However, it potentially provokes power dissipation in MIMO due to diverse transmission directions. Therefore, the beamformer architecture of a MIMO system should be well established and appropriately designed to concentrate transmission power in certain directions \cite{5}. Conventionally, there exist low-cost analog and high-performance digital beamforming techniques that can resolve the abovementioned issues \cite{7}. Analog beamforming aims to adjust the phase shifters of antennas to perform directional data transfer while consuming few circuit power resources. However, limited selection of directional beams is intrinsically induced by the quantized phases from hardware constraints. In addition, digital beamforming can perform a more comprehensive task in terms of baseband inter-beam interference cancellation, providing full management of directional beam coverage. Nevertheless, digital structure fundamentally requires a variety of RF chains connected to individual antennas, which requires substantial expenditure for configuration, power consumption and transceiver deployment. Therefore, the hybrid beamforming (HBF) architecture has become a promising and implementable solution by leveraging the advantages of both low-cost analog structures and highly directional digital structures \cite{9}. Accordingly, fewer RF chains can be constructed to associate with enormous antennas, which reduces the cost while maintaining the asymptotic performance of digital beamforming.

	In recent years, abundant research has focused on HBF architectural design but has been confined to specific structures, i.e., full connection and subconnection using separate subarrays \cite{10,11}. Due to the high degree of freedom of full connection in HBF, it enables the BS to achieve a nearly optimal solution. In \cite{14}, only analog beamforming, i.e., beam directions, is considered, which aims to maximize multi-device-based resource utilization. However, joint baseband digital and analog phase shifter precoders should be designed, which have considerably high computational complexity \cite{15}. In \cite{17,18}, beamforming schemes were designed from protocol perspectives. Although deep learning can alleviate the complexity issue, it cannot derive either closed forms or proofs of optimality or convergence \cite{21}. Accordingly, some papers focused on designs of a disjoint solution of analog and digital beamforming under the conventional fully connected HBF architecture \cite{22,23}. For separate subarrays in HBF, RF chains take control of independent subsets of phase shifters and antenna elements. In surveys \cite{10,11}, the authors investigated many subconnection architectures of HBF. In \cite{24}, the authors design a disjoint analog/digital beamforming for subarray architectures. Nonetheless, the performance of subarrays is potentially limited due to the inflexibility and inaccessibility of information from other RF chains or antenna sets. In \cite{newb1}, low-resolution quantizers of digital-to-analog converters were adopted in partially-connected HBF architecture for improving system spectral efficiency. In \cite{newb2}, partially-connected HBF is designed for advanced applications for multiuser integrated sensing and communication systems. In \cite{28}, the bit error rate was minimized under various RF/antenna connections. Therefore, it is important to design a novel HBF architecture to achieve high operational flexibility and asymptotic performance of digital beamforming. 
	
	However, increasingly large-scale HBF architectures with enormous antennas and RF chains serving multiple devices will lead to substantially high overhead in terms of power and complexity, which has almost not been discussed or optimized in the above mentioned works \cite{14,15,17,18,21,22,23,24,28}. In \cite{27}, the authors discussed different RF connections for subarray-based hybrid beamformers. The power consumption factor under each architecture was also analyzed. In \cite{new1}, a generalized subarray-connected architecture was introduced for improving the EE in HBF, which aims to solve rate optimization problems with each subproblem related to a single subarray. A dynamic adjustment of the subarray architecture under partial connections was designed to achieve high EE performance \cite{new1}. In \cite{29,30,31,32}, the authors addressed the importance of optimizing EE from either an architectural or hardware perspective. The operating power amplifiers and transmit power were considered in \cite{29,30} but without consideration of RF/antenna effects, which were taken into account in \cite{31,32}. However, all of these studies focused on fully connected HBF architectures optimizing EE, which requires high complexity. To elaborate further, conventional HBF architecture with high power consumption may be inappropriate for the future promising applications in joint radar sensing and communication systems \cite{newisac1, newisac2}. Moreover, under such architectures, it is unnecessary for all RF chains and antennas to be operated due to worse channel conditions and high power consumption, which should be appropriately selected to prevent performance degradation. Moreover, in \cite{newnew1}, joint antenna and beamforming is designed in automotive radar sensing-communications for mitigating beam interference under certain power consumption. In \cite{newnew2}, an antenna selection scheme in HBF for terahertz-band sensing-communication is considered for maximizing the spectrum efficiency of the overall system. In \cite{33}, antenna selection methods were designed for full-connection HBF maximizing the overall data rate. On the other hand, \cite{36,37} performed RF chain selection in a fully connected HBF architecture by optimizing energy utilization. However, none of these studies considered the joint selection of RF chains or antennas, which is capable of supporting potentially resilient beamforming with high EE. Accordingly, the joint design of RF/antenna selection and power consumption, as well as flexible HBF architectures, should be considered.

	In a flexible HBF architecture, RF chains are partially connected to the corresponding antenna sets. However, some direct links between RF chains and antennas do not exist, which potentially induces information vanishing. Message passing (MP), which originated from factor graph theory \cite{38}, is deemed a prospective technique for large-scale networking, which possesses the capability to transfer desirable data\textsuperscript{\ref{note1}}\footnotetext[1]{Note that the term "data" in MP originates from the field of computer science, which in this paper is referred to as the candidate solution of parameters in our system, i.e., parameters of RF/antenna selections are passed among neighboring nodes. The term "data" in the wireless communication domain specifically indicates transferred "uplink/downlink data" resulting in a corresponding uplink/downlink data rate, which is conveyed from a transmitter and received by a receiver end.\label{note1}} and information among enormous links and nodes \cite{39, yang1}. In \cite{42}, MP-aided subarray-based link optimization was resolved, whereas the power allocation problem for beamforming was considered by MP in \cite{43}. In \cite{44}, MP-based receiving beamforming was designed under full connection, while antenna selection using MP was developed in \cite{45}. However, the joint design of a partially connected HBF architecture and RF/antenna selection has not been studied in previous works. Benefitting from resilient HBF, MP is particularly suitable for partially connected HBF architectures with indirect links. The policy and determination of RF sets can be readily passed through HBF links to the corresponding antenna sets. However, under a flexible HBF connection architecture, partial links potentially lead to unreachable information between faraway nodes, and it becomes compellingly imperative to design a novel HBF architecture. 

Inspired by the concept of low-density parity-check (LDPC) \cite{46} coding, the HBF architecture is redesigned to ensure that different sets become dependent with the fewest links \cite{47}. 
In other words, information can be definitely conveyed from one set to the other one within polynomial time, which is not considered in existing HBF architecture design. Moreover, under the partial connections of LDPC coding, MP can effectively and efficiently convey policies among RF/antenna nodes in HBF. Accordingly, we design a novel flexible HBF architecture by employing an LDPC-based structure considering the minimization of power consumption and RF/antenna selection in a massive MIMO communication system. The contributions of this work are summarized as follows.
\begin{description}
	\item[$\bullet$]
	We design an LDPC-based connection for an HBF structure considering antenna and RF chain selection at the receiver. Compared with a fully connected beamformer, each RF chain is partially connected to a cluster of antennas. For arbitrary pairs of RF chains or antennas, we prove that there exists at least one path that can transfer information from one set to the other within polynomial time, i.e., all RF/antenna sets at the receiver are dependent on the LDPC-based connection, unlike in existing partially connected beamformers.
	
	\item[$\bullet$]
	We simultaneously consider both antenna and RF chain selections at the receiver as well as hybrid beamforming, which have not been jointly designed in the literature. We aim to minimize the receiver power consumption constrained by the minimum transmission data rate and the maximum tolerable system power usage. The formulated problem is decomposed into subproblems of selection of antennas and RF chains as well as hybrid beamforming at the receiver, which are computed separately in each virtual processing controller depending on the exchanged message.
		
	\item[$\bullet$]
	We propose a message-passing enhanced antenna and RF chain selection (MARS) scheme to minimize the operational circuit power of antennas and RF chains as well as a heuristic genetic algorithm for receiving hybrid beamforming. Sequential and parallel MP schemes for MARS are designed, namely, MARS-S and MARS-P, respectively. The antenna/RF virtual controller employing MARS-S passes the determined solutions in a sequential manner, while candidate outcomes of MARS-P are simultaneously transmitted to neighboring RF/antenna nodes. Considering complexity and convergence, MARS can be performed in either a centralized or distributed manner. In a centralized architecture, more RF/antenna nodes are clustered and managed by the virtual controller than a distributed system. Moreover, genetic-based HBF comprises genetic process of gene generation initialization, elite selection, crossover, and mutation.
	
	\item[$\bullet$]
    We evaluate the system performance of MARS employing LDPC-based connections in receiving HBF under both centralized and distributed architectures. We quantify the power consumption and corresponding EE with respect to different quality-of-service (QoS) values, numbers of clustered RF/antenna nodes and connections, and insertion loss values between RF chains and antenna elements. The proposed MARS scheme outperforms existing techniques by realizing lower power consumption and higher EE.
\end{description}

\textit{Notations}: We define bold capital letter $\mathbf{A}$ as a matrix, bold lowercase letter $\mathbf{a}$ as a vector, and $\mathcal{A}$ as a set. $[\mathbf{A}]_{mn}$ denotes the $(m,n)$-th element of matrix $\mathbf{A}$. Matrix operations $\mathbf{A}^H$, $\mathbf{A}^T$ and $\mathbf{A}^{-1}$ represent the Hermitian transpose, transpose and inverse of $\mathbf{A}$. $\lVert \mathbf{A} \rVert_F$ indicates the Frobenius norm of $\mathbf{A}$. $\det(\cdot)$, $\diag(\cdot)$ and $\mathbbm{E}[\cdot]$ are defined as the determinant, diagonal matrix and expectation operation, respectively. $\mathbf{A} \circ \mathbf{B}$ indicates the element-wise product of matrices $\mathbf{A}$ and $\mathbf{B}$. $\mathbbm{1}(\mathcal{X})$ denotes an indicator function that is one when event $\mathcal{X}$ takes place. The operations $\neg$, $\land$, $\lor$ and $\oplus$ indicate the binary-wise logical operations NOT, AND, OR and exclusive-OR, respectively. Additionally, $\bigvee$ denotes the sequential logic operations of $\lor$, i.e., $\bigvee_{i=1}^{N_x}x_i= x_1 \lor x_2 ... \lor x_{N_x}$, where $N_x$ is the dimension of $\{x_i\}$. $\text{mod}(a,b)$ is the modulus operation with arbitrary integers $a$ and $b$, e.g., $\text{mod}(a,b)=0$ indicates that $a$ is an integer multiple of $b$. $\left\lceil \cdot \right\rceil$ is a ceiling function. $\mathfrak{R}\{A\}$ and $\mathfrak{I}\{A\}$ indicate the real and imaginary parts of a complex variable $A$, respectively.

\begin{figure}
\centering
\includegraphics[width=3.4in]{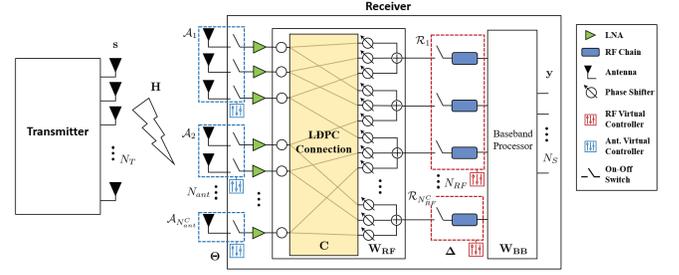}
\caption{Proposed structure of the antenna and RF chain selection for hybrid beamforming in uplink.} \label{system}
\end{figure}

\section{System Model and Problem Formulation}
\subsection{System Model} \label{SM}
    We consider a partially connected HBF architecture at the receiving BS in an uplink massive MIMO communication system, as shown in Fig. \ref{system}. The receiving BS is equipped with RF chains $\mathcal{R} \in \{1,...,n,...,N_{RF}\}$ and receiving antenna elements $\mathcal{A} \in \{1,...,m,...,N_{ant}\}$, where $N_{RF}$ and $N_{ant}$ indicate the total numbers of HBF RF chains and receiving antennas, respectively. Without loss of generality, we assume that the hybrid beamformer possesses fewer RF chains than antennas due to the reduction in hardware expenditure in the fully digital beamformer, i.e., $N_{RF} \leq N_{ant}$. The receiving hybrid beamformer consists of an analog beamformer $\mathbf{{W}_{RF}} \in {\mathbbm{C}}^{N_{RF}\times N_{ant}}$ and a baseband digital beamformer $\mathbf{{W}_{BB}} \in {\mathbbm{C}}^{N_{S}\times N_{RF}}$, where $N_{S}$ denotes the number of data streams. The connection $\mathbf{C} \in {\mathbbm{Z}}^{N_{RF}\times N_{ant}}$ between the analog and baseband beamformers can be either in a full or partial connection, where ${[\mathbf{C}]}_{nm} \in \{0,1\}, \forall n,m,$ indicates whether the $n$-th RF chain is linked to the $m$-th antenna. The desired transmitted signal is defined as $\mathbf{s} \in {\mathbb{C}}^{N_{T}\times 1}$, where $N_{T}$ is the number of antennas at the transmitter of the user terminal.

To manage the energy consumption of HBF, we can switch off spoiled RF chains or antennas to optimize the power utilization. The selection matrices of RF chains and antenna elements are denoted by $\mathbf{\Delta} = \diag(\boldsymbol{\delta}) = \diag(\delta_1,...,\delta_{n},...,\delta_{N_{RF}}) $ and $\mathbf{\Theta} = \diag(\boldsymbol{\theta}) = \diag(\theta_1,...,\theta_{m},...,\theta_{N_{ant}})$, respectively, where $\delta_{n}\in\{0,1\}$ and $\theta_{m}\in\{0,1\}$ indicate the selection decisions for RF chains and antennas, e.g., $\delta_{n}=1$ and $\theta_{m}=1$ mean that the $n$-th RF chain and $m$-th antenna are selected to be turned on. Since we focus on the designs at the receiving hybrid beamformer, the transmission power $P_{T}$ is considered to be equally allocated to transmit antennas\textsuperscript{\ref{noten}}\footnotetext[2]{Notation of $\mathbf{s}$ can be similarly generated with the transmit beamforming technique, as considered in most papers, i.e., $\mathbf{s}$ can be expressed in the form of $\mathbf{s} = \mathbf{W_{RF,tx}} \mathbf{W_{BB,tx}} \mathbf{x}$, where $\mathbf{W_{RF,tx}}, \mathbf{W_{BB,tx}}$ and $\mathbf{x}$ are defined as the transmit RF, baseband beamformers, and original transmit signal, respectively. However, this will require highly complex algorithms with unaffordable complexity order when jointly considering both transmit and receiving beamforming as well as RF/antenna selection mechanisms. \label{noten}}, i.e., $\mathbbm{E}\left[\mathbf{s}{\mathbf{s}}^H\right] \!=\! \frac{P_{T}}{N_T}\mathbf{I}_{N_{T}}$, where $\mathbf{I}_{N_{T}}$ is the identity matrix of size $N_{T}$. The received uplink signal is given by
\begin{align} \label{recvsignal}
    \mathbf{y} &= \sqrt{\rho \left(1-\beta\right)} \mathbf{W_{BB} \Delta \left(C \circ W_{RF} \right) \Theta H s} \notag \\
    & \qquad\qquad + \mathbf{W_{BB} \left( C \circ W_{RF}\right) n},
\end{align}
where $\circ$ is the Hadamard product for elementwise multiplication. In $\eqref{recvsignal}$, $\rho$ is the distance-based loss and $\mathbf{H} \in \mathbb{C}^{N_{ant} \times N_T}$ is the small-scale fading channel, which is defined based on the Saleh-Valenzuela model \cite{24,30} as
\begin{align} \label{SVmodel}
    \mathbf{H} = \sqrt{\frac{N_T N_{ant}}{L}}\sum_{\ell=1}^{L} {\alpha}_{\ell} \cdot \boldsymbol{\alpha}_{r}({\phi}_{\ell}^{r}) \boldsymbol{\alpha}_{t}^{H}({\phi}_{\ell}^{t}),
\end{align}
where $L$ is the number of multipaths and ${\alpha}_{\ell}  \sim \mathcal{CN}(0,1)$ is the complex channel gain of the $\ell$-th path. The array response vector corresponding to the angle of arrival (AoA) of the HBF receiver and the angle of departure (AoD) of the transmit BS are given by $\boldsymbol{\alpha}_{r}({\phi}_{\ell}^{r})$ and $\boldsymbol{\alpha}_{t}({\phi}_{\ell}^{t})$, respectively, along with their incident receiving and transmit angles of ${\phi}_{\ell}^{r}$ and ${\phi}_{\ell}^{t}$. In $\eqref{recvsignal}$, $\mathbf{n} \!\sim\! \mathcal{CN}(0,N_{0}\mathbf{I}_{N_R})$ is complex additive white Gaussian noise with a noise power spectral density of $N_{0}$. Furthermore, we consider the insertion loss $\beta$ in the uplink receiving HBF-based BS \cite{48}, which is the connection loss between an RF chain and its connected antenna elements due to signal power dissipation. Based on the received uplink signal in $\eqref{recvsignal}$, the achievable system rate can be written as
\begin{align} \label{recvcapacity}
	 &R(\mathbf{\Delta,\Theta}, \mathbf{W}_{\mathbf{RF}}, \mathbf{W}_{\mathbf{BB}}) \!=\! \notag\\
	 &\log_2 \left[\det\left| \mathbf{I}_{N_S} \!+\! \frac{P_{T} \rho \left(1-\beta\right)} {N_{0} N_{T}}\mathbf{W_{BB} \Delta \left( C \circ W_{RF}\right) \Theta H} \mathbf{H}^{H}  \right. \right. \nonumber\\
	& \left. \left. \mathbf{\Theta}^{H} \mathbf{W}^{H}_{\mathbf{RF}}\mathbf{\Delta}^{H} \mathbf{W}^{H}_{\mathbf{BB}} \left(\mathbf{W_{BB} \left( C \circ W_{RF}\right)} \mathbf{W}^{H}_{\mathbf{RF}} \mathbf{W}^{H}_{\mathbf{BB}}\right)^{-1}\right|\right].
\end{align}
The total power consumption at the receiver HBF-based BS can be represented by \cite{30,32}
\begin{align} \label{powerconsumption}
	 P(\boldsymbol{\delta},\boldsymbol{\theta}) 
	 & = P_{BB} + \sum_{n=1}^{N_{RF}}  \left[ {\delta}_{n} \left( P_{RF}+P_{ADC} \right) \right. \notag \\
	 & \left. + \sum_{m=1}^{N_{ant}} \left( {\theta}_{m} P_{LNA} + [\mathbf{C}]_{nm} {\delta}_{n} {\theta}_{m} P_{PS} \right) \right],
\end{align}
where $P_{BB}$, $P_{RF}$, $P_{ADC}$, $P_{LNA}$ and $P_{PS}$ denote the power of HBF baseband signal processing, circuit operating power of RF chains, analog-to-digital converter, low-noise power amplifiers (LNA) and phase shifters at receiving antennas, respectively. We observe that $\eqref{powerconsumption}$ includes a static power term $P_{BB}$ and dynamic power terms $P_{RF}$, $P_{ADC}$, $P_{LNA}$ and $P_{PS}$ because the dynamic terms are determined by the selection of RF chains and antennas. The system parameters and corresponding notations are listed in Table \ref{Parameters}.

\begin{table}
\begin{center}
\scriptsize
\caption {Definition of System Parameters}
    \begin{tabular}{|l|l|}
    \hline
        Parameter & Notation \\
        \hline\hline
		Sets of antennas and RF chains & $\mathcal{A}, \mathcal{R}$ \\ \hline
		Numbers of transmitting/receiving antennas and RF chains & $N_{T}, N_{ant}, N_{RF}$ \\ \hline
		Number of data streams & $N_{S}$  \\ \hline
		HBF analog and baseband beamformers & $\mathbf{{W}_{RF}}, \mathbf{{W}_{BB}}$ \\ \hline
		HBF connection & $\mathbf{C}$ \\ \hline
		RF chain and antenna selection matrices of & $\mathbf{\Delta}, \mathbf{\Theta}$ \\ \hline
		RF chain and antenna element-wise selection & $\delta_{n}, \theta_{m}$ \\ \hline
		Transmit, received, and noise signals & $\mathbf{s}, \mathbf{y}, \mathbf{n}$ \\ \hline
		Transmit power & $P_{T}$ \\ \hline
		Noise power & $N_{0}$ \\ \hline
		Small- and large-scale channel fading & $\mathbf{H}, \rho$ \\ \hline
		Number of multipaths & $L$ \\ \hline
		Insertion loss & $\beta$ \\ \hline
		Achievable rate & $R(\mathbf{\Delta,\Theta}, \mathbf{W}_{\mathbf{R\!F}}, \mathbf{W}_{\mathbf{B\!B}})$ \\ \hline
		Total power consumption & $P(\boldsymbol{\delta}, \boldsymbol{\theta})$ \\ \hline
		\tabincell{l}{Power consumption of baseband processing, RF chains, \\ ADC, LNA circuit, and phase shifter} & \tabincell{l}{$P_{B\!B},  P_{R\!F}$, $P_{A\!D\!C}$\\$ P_{L\!N\!A},  P_{P\!S}$} \\ \hline
		Number of LDPC connections & $N_{Conn}$ \\ \hline
		Sets of RF and antenna controllers & $\mathcal{R}^{\mathcal{C}} ,  \mathcal{A}^{\mathcal{C}}$ \\ \hline
		$k$-th RF and $l$-th antenna controller & $\mathcal{R}_{k},\mathcal{A}_{l}$ \\ \hline
		Numbers of RF and antenna controllers & $N_{RF}^{C} ,  N_{ant}^{C}$ \\ \hline
		Numbers of RF and antenna elements controlled & $N_{RF,k} ,  N_{ant,l}$ \\ \hline
		RF decision/antenna update sets of RF controllers & $\boldsymbol{\delta}_{\mathcal{R}_{k}} ,  \boldsymbol{\theta}_{\mathcal{R}_{k},l}$ \\ \hline
		Antenna decision/RF update sets of antenna controllers & $\boldsymbol{\theta}_{\mathcal{A}_{l}} ,  \boldsymbol{\delta}_{\mathcal{A}_{l},k}$ \\ \hline
		MP operation from RF to antenna controller & $\mu_{\mathcal{R}_{k} \rightarrow \mathcal{A}_{l}}(\cdot)$ \\ \hline
		MP operation from antenna to RF controller & $\nu_{\mathcal{A}_{l} \rightarrow \mathcal{R}_{k}}(\cdot)$ \\ \hline
		Thresholds for RF and antenna selection & $\eta_{r}, \eta_{a}$ \\ \hline
    \end{tabular} \label{Parameters}
\end{center}
\end{table}

\subsection{Problem Formulation} \label{FS}
	Benefiting from a partial HBF architecture that employs LDPC-based connections\textsuperscript{\ref{note3}}\footnotetext[3]{LDPC-based connections still require all antennas and RF chains for channel estimation and selection mechanism. Note that new hardware architecture may increase hardware cost as well as impose higher insertion and routing losses from additional switches. At high frequencies in millimeter wave, there are typically high losses in printed circuit boards and switches. \label{note3}}, we can guarantee the dependency of the decided information, i.e., information can be exchanged among arbitrary nodes of RF chains and antennas. In this paper, we aim to minimize the total operational circuit power $P(\boldsymbol{\delta},\boldsymbol{\theta})$ by selecting the candidate RF chains of $\boldsymbol{\delta}$ and antenna elements of $\boldsymbol{\theta}$ at the uplink receiving HBF-enabled BS as well as the hybrid beamformers $\mathbf{W}_{\mathbf{RF}}$ and $\mathbf{W}_{\mathbf{BB}}$, which is formulated as
\begingroup
\allowdisplaybreaks
\begin{subequations} \label{mainprob}
  \begin{align}
        &\min_{\boldsymbol{\delta},\boldsymbol{\theta}, \mathbf{W}_{\mathbf{RF}}, \mathbf{W}_{\mathbf{BB}}}\quad {P(\boldsymbol{\delta},\boldsymbol{\theta})}, \\
        &\text{s.t.} \quad {\delta}_{n} \in \{0,1\}, \quad \forall n, \label{RF_blim} \\
        &\quad\quad {\theta}_{m} \in \{0,1\}, \quad \forall m, \label{ant_blim} \\
        &\quad\quad {N}_{S} \leq \sum_{n=1}^{N_{RF}} {\delta}_{n} \leq \sum_{m=1}^{N_{ant}} {\theta}_{m}, \label{RF_numlim} \\
        &\quad\quad R(\mathbf{\Delta,\Theta},\mathbf{W}_{\mathbf{RF}}, \mathbf{W}_{\mathbf{BB}}) \geq R_{req}, \label{SR_req} \\
        &\quad\quad P(\boldsymbol{\delta},\boldsymbol{\theta}) \leq P_{max}, \label{P_lim} \\
        &\quad\quad P_{PS}\sum_{m=1}^{N_{ant}} {\theta}_{m}[\mathbf{C}]_{nm} \leq ( 1 \!-\!\beta)P_{o}, \quad \forall n, \label{RF_insertionloss}\\
        &\quad\quad |[\mathbf{W}_{\mathbf{RF}}]_{nm}|=1, \quad \forall m, n \label{con_RF}.
  \end{align}
\end{subequations}
\endgroup
Constraints $\eqref{RF_blim}$ and $\eqref{ant_blim}$ represent the binary selection decisions for RF chains and antenna elements, respectively. Constraint $\eqref{RF_numlim}$ indicates the limitation of the HBF architecture that the number of data streams is fewer than the number of operating RF chains, while the number of selected RF chains is generally smaller than the number of operating antennas. $\eqref{SR_req}$ guarantees that the minimum QoS requirement of $R_{req}$ is satisfied at the HBF-enabled receiving BS. Constraint $\eqref{P_lim}$ limits the system power to the maximum allowable system power $P_{max}$. In $\eqref{RF_insertionloss}$, we consider that each RF chain can support a maximum power of $P_{o}$ associated with the maximum number of connected receiving antenna elements. The circuit insertion loss $\beta$ is also considered at the receiver BS. Constraint $\eqref{con_RF}$ is a modulus constraint for the analog beamformer. We can observe that the beamformer designs of $\mathbf{W_{RF}}$ and $\mathbf{W_{BB}}$ will induce a complex problem. We observe that problem $\eqref{mainprob}$ is complex owing to the coupled variables of continuous beamformers and the discrete RF/antenna selection parameters. Therefore, we design an MP-empowered RF/antenna selection scheme under an LDPC-based HBF connection, which is followed by the heuristic beamformer scheme.


\section{Design of an LDPC-based HBF Connection}
	Inspired by the parity-check property in coding theory, we design an LDPC-based HBF connection $\mathbf{C}$ between RF chains and antennas, which is partially linked with the following properties \cite{46}:
	\begingroup
	\allowdisplaybreaks
\begin{subequations} \label{LDPC}
  \begin{align}
       &\sum_{m=1}^{N_{ant}} {[\mathbf{C}]}_{nm} = {N}_{Conn} \overset{\underset{\mathrm{(a)}}{}}{>} \left\lceil \frac{N_{ant}}{N_{RF}} \right\rceil, \quad \forall n , \label{LDPC1}\\
       &\sum_{n=1}^{N_{RF}} {[\mathbf{C}]}_{nm} \geq 1, \quad \forall m, \label{LDPC2} \\
       & \sum_{n=1}^{N_{RF}} \sum_{ n'=n+1}^{N_{RF}} \mathbbm{1}\left( \chi_{nn'}\geq 1 \right) \geq N_{RF}-1, \label{LDPC3}
  \end{align}
\end{subequations}
\endgroup
where ${N}_{Conn}$ indicates the number of antennas connected to an RF chain and $\chi_{nn'}=\sum_{m=1}^{N_{ant}} \mathbbm{1}\Big( \left( [\mathbf{C}]_{nm}+[\mathbf{C}]_{n'm} \right) \geq 2\Big)$ denotes the number of commonly paired antenna nodes. The ceiling operation $\left\lceil \cdot \right\rceil$ prevents noninteger values when $\varrho=\text{mod}(N_{ant},N_{RF})\neq 0$, and inequality (a) in $\eqref{LDPC1}$ guarantees the LDPC property that there exists at least a direct or an indirect link between two arbitrary nodes of RF chains or antennas. Notably, constraint $\eqref{LDPC3}$ is different from the original parity-check mechanism, which ensures interconnection among RF/antenna nodes, as proven in the following lemma.
\begin{definition} \label{Def1}
Given $x$ nodes, we only require $(x-1)$ edges to establish connections among them. The worst case is that nodes $1$ and $x$ are capable of communicating via the maximum $(x-1)$ hops, while the best case with the minimum distance is a single hop between any of the neighboring nodes.
\end{definition}
\begin{lemma} \label{lm1}
	Considering the LDPC-based connection scheme in $\eqref{LDPC}$, there exists at least a direct or an indirect link between two arbitrary nodes of RF chains or antennas.
\end{lemma}
\begin{proof}
We start from ${N}_{Conn}=\left\lceil\frac{N_{ant}}{N_{RF}}\right\rceil$, which implies that at least independent links can be constructed considering constraint $\eqref{LDPC2}$, e.g., ${[\mathbf{C}]}_{nm}=1$ when $(n-1)\cdot \left\lceil\frac{N_{ant}}{N_{RF}}\right\rceil + 1\leq m\leq n\cdot \left\lceil\frac{N_{ant}}{N_{RF}}\right\rceil$. 
Now, we provide an additional link as ${N}_{Conn}=\left\lceil\frac{N_{ant}}{N_{RF}}\right\rceil+1$, which indicates that there will exist $N_{ant}-\left\lceil\frac{N_{ant}}{N_{RF}}\right\rceil$ antenna selection cases. Therefore, according to $\eqref{LDPC3}$, any two arbitrary RF chains should have at least a single common connected antenna, i.e., $\mathbbm{1}\Big( \left(\mathbf{C}_{nm}+\mathbf{C}_{n'm}\right)\geq 2\Big)$ when $n\neq n'$. Moreover, based on Definition \ref{Def1} and the last inequality of $\eqref{LDPC3}$, we are able to guarantee that there exists at least a direct or an indirect link between two arbitrary nodes of RF chains or antennas. This completes the proof.
\end{proof}
Based on Lemma \ref{lm1}, we observe from $\eqref{LDPC}$ that higher spatial diversity is attained when ${N}_{Conn}$ grows and achieves the full connection architecture if ${N}_{Conn}=N_{ant}$ holds. In contrast, when $N_{Conn} = \left\lceil \frac{N_{ant}}{N_{RF}} \right\rceil$, we have a basic architecture of the partially connected method. Compared to the fully connected structure of HBF, the designed LDPC-based HBF architecture can provide higher flexibility and more degree of freedom. Moreover, under the established architecture, all nodes of RF chains and antennas will be dependent according to Lemma \ref{lm1}, i.e., information can be conveyed from one node to another within finite time. The construction algorithm of the proposed LDPC-based HBF connection\textsuperscript{\ref{note4}}\footnotetext[4]{We can implement the proposed LDPC-based connection by initializing a configuration of fully connected hybrid beamformers. After obtaining the optimal LDPC connection in Algorithm \ref{alg:LDPC}, we are able to permanently turn off the unused links between RF chains and antennas, leaving the remaining elements operable. To elaborate a little further, the LDPC connection can be regarded as a long-term mechanism compared to the short-term element selection dealing with a small-scale fading channel. Leveraging such a scheme requires further complex analysis and design, which is outside the scope of this paper and can be left as future promising work.\label{note4}} between RF chains and antennas is demonstrated in Algorithm \ref{alg:LDPC}. In Fig. \ref{example_ld}, we provide an example of the construction of the proposed LDPC-based structure for the MP scheme. As shown in Fig. \ref{ld1} following Algorithm \ref{alg:LDPC}, we first establish $N_{Conn}-1=2$ links (red solid lines) between RF chains and antenna nodes to satisfy $\eqref{LDPC2}$, while an additional link (blue dotted lines) per RF chain satisfying $\eqref{LDPC1}$ for $N_{Conn}=3$ is also established. The constraint $\eqref{LDPC3}$ is also guaranteed with $\chi_{12}=2$, $\chi_{23}=1$, and $\chi_{13}=1$, i.e., $\sum_{n=1}^{3} \sum_{ n'=n+1}^{3} \mathbbm{1}\left( \chi_{nn'}\geq 1 \right) = 3$. Therefore, messages can be conveyed to arbitrary nodes (green arrows). On the other hand, if the number of antennas is an integer multiple of the number of RF chains, as exemplified in Fig. \ref{ld2}, an independent association between arbitrary nodes will be induced in the first-phase construction. For example, a message from antenna 6 cannot be passed to antenna 4 via red paths. To address this, additional blue paths obeying $\eqref{LDPC3}$ provide connections among arbitrary nodes, e.g., messages between antennas 4 and 6 can therefore be perfectly delivered but with more hops.

\begin{figure}
    \centering
    \subfigure[]{\includegraphics[width=1.4 in]{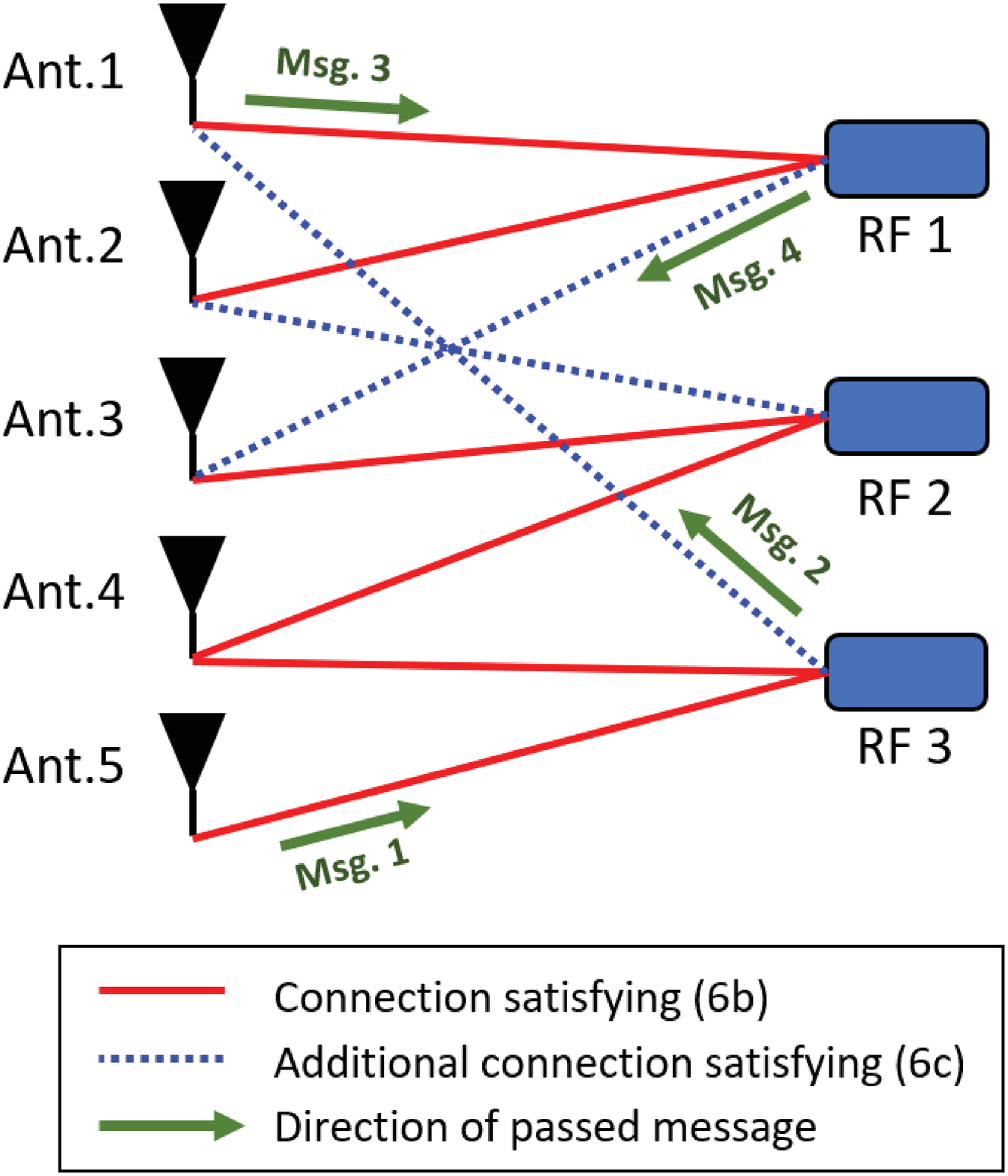} \label{ld1}}
    \subfigure[]{\includegraphics[width=2 in]{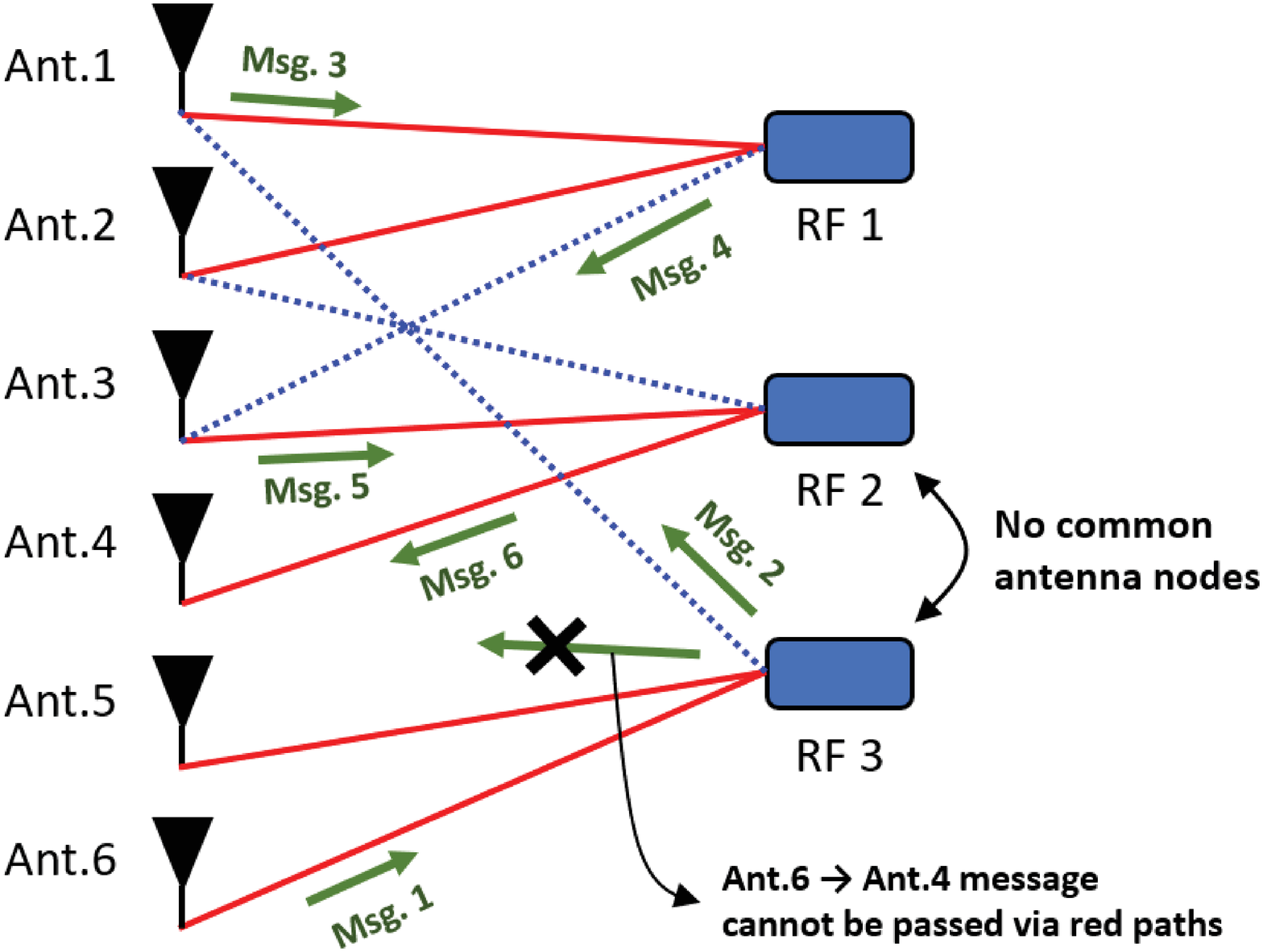} \label{ld2}}
    \caption{Example of the proposed LDPC-based structure for MP considering $\varrho=\text{mod} (N_{ant}, N_{RF})$ with $N_{ant}\in\{5,6\}$ antenna nodes and $N_{RF}=3$ RF chain nodes. (a) $\varrho\neq 0$ and $N_{ant}=5$; (b) $\varrho=0$ and $N_{ant}=6$. Ant. and Msg. are acronyms for antenna and message, respectively.} \label{example_ld}
\end{figure}   

	To manage RF chains and antenna elements in the receiving BS, we consider that the HBF architecture is divided into several distributed virtual controllers as RF and antenna controllers. That is, several RF chains and antennas are clustered and controlled by their corresponding controllers, which are depicted in Fig. \ref{system}. We use the term controller to refer to a virtual controller in the remaining content for simplicity. In the LDPC-based connected HBF architecture, we deploy $N_{RF}^{C}$ RF chain controllers as $ \mathcal{R}^{\mathcal{C}} \in \{ \mathcal{R}_{1},...,\mathcal{R}_{k},...,{\mathcal{R}}_{N_{RF}^{C}}\}$ and $N_{ant}^{C}$ antenna controllers as $\mathcal{A}^{\mathcal{C}} \in \{\mathcal{A}_{1},...,\mathcal{A}_{l},...,{\mathcal{A}}_{N_{ant}^{C}}\}$, where ${\mathcal{R}}_{k}$ and ${\mathcal{A}}_{l}$ denote the control units of the $k$-th RF chain controller and the $l$-th antenna controller, respectively. Without loss of generality, we assume that antennas and RF chains cannot be managed by more than one controller, i.e., ${\mathcal{R}}_{k} \cap  {\mathcal{R}}_{k'} \!=\! \emptyset, \forall k' \!\neq\! k,$ and ${\mathcal{A}}_{l} \cap  {\mathcal{A}}_{l'} \!=\! \emptyset, \forall l'\!\neq\! l$. We assume that RF controller $\mathcal{R}_k$ manages $N_{RF,k}$ RF chains and antenna controller ${\mathcal{A}}_{l}$ manages $N_{ant,l}$ antenna elements. Due to the designed partial HBF structure, each control unit obtains only partial knowledge from neighboring linked nodes and its own determination. Therefore, we define that the $k$-th RF chain controller possesses an RF decision set $\boldsymbol{\delta}_{\mathcal{R}_k}\!=\! \{ \delta_{\mathcal{R}_k,1},...,\delta_{\mathcal{R}_k,N_{RF,k}} \}$ and an antenna decision set from the $l$-th connected antenna controller $\boldsymbol{\theta}_{\mathcal{R}_k, l}\!=\! \{ \theta_{\mathcal{R}_k,l,1},...,\theta_{\mathcal{R}_k,l,N_{ant,l}} \}$. Similarly, the $l$-th antenna controller has its own antenna decision set $\boldsymbol{\theta}_{\mathcal{A}_l} \!=\! \{ \theta_{\mathcal{A}_l,1},...,\theta_{\mathcal{A}_l,N_{ant,l}} \}$ and RF selection decision set from the neighboring RF controller $k$, namely, $\boldsymbol{\delta}_{\mathcal{A}_l,k}\!=\! \{ \delta_{\mathcal{A}_l,k,1},...,\delta_{\mathcal{A}_l,k,N_{RF,k}} \}$. The policy subsets at each controller are included in the total sets of RF and antenna selections as $\{\boldsymbol{\delta}_{\mathcal{R}_k}, \boldsymbol{\delta}_{\mathcal{A}_l, k}, \forall k,l\} \in \boldsymbol{\delta}$ and $\{\boldsymbol{\theta}_{\mathcal{A}_l}, \boldsymbol{\theta}_{\mathcal{R}_k,l}, \forall k,l\} \in \boldsymbol{\theta}$, respectively. 

\begin{algorithm}[!tb]
\scriptsize
\caption{LDPC-based HBF connection}
\SetAlgoLined
\DontPrintSemicolon
\label{alg:LDPC}
\begin{algorithmic}[1]
\STATE Initialization: $\mathcal{A}$, $\mathcal{R}$, $N_{Conn}$\\
\STATE \textbf{(Independent Connection Establishment)}
\STATE A temporary non-established connection set is selected as $\mathcal{M}=\mathcal{A}$
\FOR{$n =1,...,N_{RF}$}
		\STATE A subset $\mathcal{M}'$ with $\left\lceil \frac{N_{ant}}{N_{RF}} \right\rceil$ elements from $\mathcal{M}$ is randomly selected
		\STATE A binary vector indexed by subset $\mathcal{M}'$ is generated based on $\eqref{LDPC1}$, i.e., we have $[\mathbf{C}]_{nm}=1$ for $m\in \mathcal{M}'$ with the remaining elements being zero for $m \notin \mathcal{M}'$
		\STATE The non-established connection set is updated as $\mathcal{M} \leftarrow \mathcal{M} - \mathcal{M}'$
\ENDFOR
\STATE \textbf{(Node Dependency Construction)}
	\STATE An initial pairing RF candidate $n\in\mathcal{R}$ is randomly selected
	\STATE A temporary non-selected RF set is selected as $\mathcal{N}=\mathcal{R}\backslash \{n\}$ 
\FOR{$iter =1,...,N_{RF}-1$}
	\STATE A pairing RF $n'\in\mathcal{N}\backslash \{n\}$ is randomly selected
	\STATE A common antenna node $m$ with  $\mathbf{C}_{n'm}=1$ is chosen, and $\mathbf{C}_{nm}=1$ is set
	\STATE The non-selected RF set is updated as $\mathcal{N}\leftarrow \mathcal{N}\backslash \{n'\}$ and $n=n'$	
\ENDFOR
\end{algorithmic}
\end{algorithm}

\section{Proposed Message Passing Antenna and RF Chain Selection (MARS) Scheme}
	Due to its high computational complexity, the original problem in \eqref{mainprob} is decomposed into two subproblems, i.e., RF chain selection and antenna selection, which are determined by RF and antenna controllers, respectively. Each controller obtains its own solution according to the messages passed from other controllers, which can potentially reduce the system complexity. The subproblem of RF selection for the $k$-th RF chain controller is written as
\begin{subequations} \label{RFprob1}
  \begin{align}
        &\min_{\boldsymbol{\delta}_{\mathcal{R}_{k}}}\quad P \left(\boldsymbol{\delta}_{\mathcal{R}_{k}} \right) \label{RFp1}\\
        &\text{s.t.}\quad \mbox{Fixed } \boldsymbol{\theta}_{\mathcal{R}_k,l}, \boldsymbol{\delta}_{\mathcal{R}_{k'}}, \forall l,k' \not \neq k,  \label{fix1} \\ 
        &\quad\quad \eqref{RF_blim}, \eqref{RF_numlim}, \eqref{SR_req}, \eqref{P_lim}, \label{conrf}
  \end{align}
\end{subequations}
    where $\eqref{RFp1}$ indicates the power consumption of the $k$-th RF chain given by the passed information of $\boldsymbol{\theta}_{\mathcal{R}_k,l}$ from the $l$-th antenna controller and the message of $\boldsymbol{\delta}_{\mathcal{R}_{k'}}$ from the neighboring $(k')$-th RF chain in $\eqref{fix1}$. Furthermore, the subproblem of antenna selection for the $l$-th RF chain controller is given by
\begin{subequations} \label{antprob1}
  \begin{align}
        &\min_{\boldsymbol{\theta}_{\mathcal{A}_{l}}}\quad P \left(\boldsymbol{\theta}_{\mathcal{A}_{l}} \right) \label{antp1}\\
        &\text{s.t.} \quad \mbox{Fixed } \boldsymbol{\delta}_{\mathcal{A}_l,k}, \boldsymbol{\theta}_{\mathcal{A}_{l'}}, \forall k, l'\neq l, \label{fix2} \\
        &\quad\quad\eqref{ant_blim}, \eqref{RF_numlim}, \eqref{SR_req}, \eqref{P_lim}, \eqref{RF_insertionloss}, \label{conant} 
  \end{align}
\end{subequations}
where $\eqref{antp1}$ represents the total power consumption of the $l$-th antenna controller given by the passed information of $\boldsymbol{\delta}_{\mathcal{A}_l,k}$ from the $k$-th RF controller and the message of $\theta_{\mathcal{A}_{l'}}$ from the neighboring $(l')$-th antenna controller in $\eqref{fix2}$. To address the two subproblems of RF/antenna selection, we propose the MARS scheme to minimize the operational circuit power of antennas and RF chains of receiving hybrid beamformers. Two types of MARS schemes are designed with sequential and parallel message passing, denoted by MARS-S and MARS-P, respectively. MARS can be iteratively performed in either a centralized or distributed manner based on computational complexity and algorithm convergence considerations. In a centralized architecture, more RF/antenna nodes are managed by an RF/antenna controller than in a distributed system. In the following, we will elaborate the MARS schemes of MARS-S and MARS-P.

\subsection{MARS-S for Sequential Message Passing}
	In the proposed sequential mechanism of MARS-S, the antenna/RF controller passes the determined solutions in a sequential manner. As a result, the controllers are guaranteed to acquire the latest information from their neighboring nodes. The overall procedure of MARS-S is shown in Fig. \ref{exp1}, including four processing steps, which are elaborated in the following.
\subsubsection{Initialization}
	At the beginning, all RF chain and antenna controllers exchange the determined beamformer matrices of $\mathbf{W_{BB}}$ and $\mathbf{W_{RF}}$ and channel information $\mathbf{H}$ to compute the achievable rate via $\eqref{recvcapacity}$. Furthermore, all RF/antenna selection elements in $\boldsymbol{\delta}$ and $\boldsymbol{\theta}$ are randomly generated in $\{0,1\}$.
	
\subsubsection{Update Based on Received Messages}
Afterwards, the RF/antenna controllers will update the latest information on RF/antenna selections based on the received messages from the neighboring controllers. According to $\eqref{RFp1}$, the update information of the $k$-th RF chain controller at the $t$-th update consists of the $(k')$-th RF chain decision $\boldsymbol{\delta}_{\mathcal{R}_{k'}}, \forall k'\neq k$ and the passed antenna selection $\boldsymbol{\theta}_{\mathcal{R}_k,l}, \forall l$ at the $(t-1)$-th iteration, which are given by \eqref{para_update_RF} at top of next page,
\begin{figure*}
\begin{subequations} \label{para_update_RF}
    \begin{align}
	\boldsymbol{\delta}_{ {\mathcal{R}}_{k'} }^{(t)} \!=\!& \bigvee_{\mathcal{A}_l \in \mathcal{C}_{\mathcal{R}_{k}}} \left[ \Xi_{\mathcal{R}} \land  \nu_{\mathcal{A}_{l} \rightarrow \mathcal{R}_{k}} \left( \boldsymbol{\delta}_{\mathcal{A}_l,{k'}}^{(t-1)} \right)\right] \land \left\lbrace \neg \left[ \bigvee_{\mathcal{A}_l \in \mathcal{C}_{\mathcal{R}_{k}}} \left( \Xi_{\mathcal{R}} \land \boldsymbol{\delta}_{\mathcal{R}_{k'}}^{(t-1)} \right)\right] \right\rbrace, \\    
    \boldsymbol{\theta}_{\mathcal{R}_{k},l}^{(t)} \!=\!& \bigvee_{\mathcal{A}_l \in \mathcal{C}_{\mathcal{R}_{k}}} \left[ \Xi_{\mathcal{R}}^{'} \land \nu_{\mathcal{A}_{l} \rightarrow \mathcal{R}_{k}} \left( \boldsymbol{\theta}_{\mathcal{A}_{l}}^{(t-1)} \right) \right] \land \left\lbrace \neg \left[ \bigvee_{\mathcal{A}_l \in \mathcal{C}_{\mathcal{R}_{k}}} \left( \Xi_{\mathcal{R}}^{'} \land \boldsymbol{\theta}_{\mathcal{R}_{k},l}^{(t-1)} \right) \right]\right\rbrace,
    \end{align}
\end{subequations}
\hrulefill
\end{figure*}
where $\Xi_{\mathcal{R}}=\nu_{\mathcal{A}_{l} \rightarrow \mathcal{R}_{k}} \left( \boldsymbol{\delta}_{\mathcal{A}_l,{k'}}^{(t-1)}\right)  \oplus \boldsymbol{\delta}_{\mathcal{R}_{k'}}^{(t-1)} $ and $\Xi_{\mathcal{R}}^{'} = \nu_{\mathcal{A}_{l} \rightarrow \mathcal{R}_{k}} \left( \boldsymbol{\theta}_{\mathcal{A}_{l}}^{(t-1)}\right)  \oplus \boldsymbol{\theta}_{\mathcal{R}_{k},l}^{(t-1)}$. The connection set of the $k$-th RF controller linked to its antenna controllers is denoted by $\mathcal{C}_{\mathcal{R}_{k}}$. Furthermore, we define $\nu_{\mathcal{A}_{l} \rightarrow \mathcal{R}_{k}}(\cdot)$ as the message passing operation that directs the path from the $l$-th antenna controller to the $k$-th RF chain controller. Similarly, according to $\eqref{antp1}$, the message of the $l$-th antenna controller at the $t$-th update consists of the $(l')$-th antenna selection $\boldsymbol{\theta}_{\mathcal{A}_{l'}}, \forall l'\neq l$ and the passed RF policy $\boldsymbol{\delta}_{\mathcal{A}_l,k}, \forall k$ at the $(t-1)$-th iteration, which are obtained as \eqref{para_update_antenna} at top of next page,
\begin{figure*}
\begin{subequations} \label{para_update_antenna}
    \begin{align}
	\boldsymbol{\theta}_{ {\mathcal{A}}_{l'} }^{(t)} \!=\!& \bigvee_{\mathcal{R}_{k} \in \mathcal{C}_{\mathcal{A}_{l}}} \left[ \Xi_{\mathcal{A}} \land  \mu_{\mathcal{R}_{k} \rightarrow \mathcal{A}_{l}} \left( \boldsymbol{\theta}_{\mathcal{R}_k,{l'}}^{(t-1)} \right)\right] \land \left\lbrace \neg \left[ \bigvee_{\mathcal{R}_k \in \mathcal{C}_{\mathcal{A}_{l}}} \left( \Xi_{\mathcal{A}} \land \boldsymbol{\theta}_{\mathcal{A}_{l'}}^{(t-1)} \right)\right] \right\rbrace, \\    
    \boldsymbol{\delta}_{\mathcal{A}_{l},k}^{(t)} \!=\!& \bigvee_{\mathcal{R}_k \in \mathcal{C}_{\mathcal{A}_{l}}} \left[ \Xi_{\mathcal{A}}^{'} \land \mu_{\mathcal{R}_{k} \rightarrow \mathcal{A}_{l}} \left( \boldsymbol{\delta}_{\mathcal{R}_{k}}^{(t-1)} \right) \right] \land \left\lbrace \neg \left[ \bigvee_{\mathcal{R}_k \in \mathcal{C}_{\mathcal{A}_{l}}} \left( \Xi_{\mathcal{A}}^{'} \land \boldsymbol{\delta}_{\mathcal{A}_{l},k}^{(t-1)} \right) \right]\right\rbrace,
    \end{align}
\end{subequations}
\hrulefill
\end{figure*}
where $\Xi_{\mathcal{A}}=\mu_{\mathcal{R}_{k} \rightarrow \mathcal{A}_{l}} \left( \boldsymbol{\theta}_{\mathcal{R}_k,{l'}}^{(t-1)}\right)  \oplus \boldsymbol{\theta}_{\mathcal{A}_{l'}}^{(t-1)} $ and $\Xi_{\mathcal{A}}^{'} = \mu_{\mathcal{R}_{k} \rightarrow \mathcal{A}_{l}} \left( \boldsymbol{\delta}_{\mathcal{R}_{k}}^{(t-1)}\right)  \oplus \boldsymbol{\delta}_{\mathcal{A}_{l},k}^{(t-1)}$. The notation $\mathcal{C}_{\mathcal{A}_{l}}$ indicates the connection set of the $l$-th antenna controller associated with its RF controllers. Moreover, $\mu_{\mathcal{R}_{k} \rightarrow \mathcal{A}_{l}}(\cdot)$ is defined as the message passing operation directing the path from the $k$-th RF controller to the $l$-th antenna controller. Benefiting from the MP mechanism, we observe from the update equations $\eqref{para_update_RF}$ and $\eqref{para_update_antenna}$ that the RF/antenna controllers can readily update all information of the other RF/antenna nodes via simple logic operations to obtain their optimal solution for RF/antenna selections.

\subsubsection{Optimization}
	Based on $\eqref{para_update_RF}$ and $\eqref{para_update_antenna}$, the RF/antenna controllers solve the decomposed subproblems in $\eqref{RFprob1}$ and $\eqref{antprob1}$ to obtain the optimum RF chain and antenna selections, which are given by
\begin{align}\label{optrf}
\boldsymbol{\delta}_{\mathcal{R}_{k}}^{*} = {\arg\!\min}_{\boldsymbol{\delta}_{\mathcal{R}_{k}}} P \left(\boldsymbol{\delta}_{\mathcal{R}_{k}} \right) \mbox{ with fixed  } \boldsymbol{\theta}_{\mathcal{R}_k,l}, \boldsymbol{\delta}_{\mathcal{R}_{k'}}, \forall l,k' \not \neq k,
\end{align}
and
\begin{align}\label{optant}
\boldsymbol{\theta}_{\mathcal{A}_{l}}^{*} = {\arg\!\min}_{\boldsymbol{\theta}_{\mathcal{A}_{l}}} P \left(\boldsymbol{\theta}_{\mathcal{A}_{l}} \right) \mbox{ with fixed  } \boldsymbol{\delta}_{\mathcal{A}_l,k}, \boldsymbol{\theta}_{\mathcal{A}_{l'}}, \forall k, l'\neq l,
\end{align}
respectively. All RF and antenna constraints should be obeyed in $\eqref{conrf}$ and $\eqref{conant}$, respectively. Due to the partial connection of the LDPC-based HBF architecture, each controller is capable of managing a small number of nodes, which can achieve a much lower complexity. Accordingly, we can readily employ a numerical brute-force method to obtain the optimal binary-based parameters, i.e., we can search all possible solutions for a group of nodes given the information passed from others. Therefore, each RF/antenna controller can obtain its respective optimal solutions for RF selection in $\eqref{optrf}$ and antenna selection in $\eqref{optant}$.

\begin{figure}
    \centering
\includegraphics[width=3.5 in]{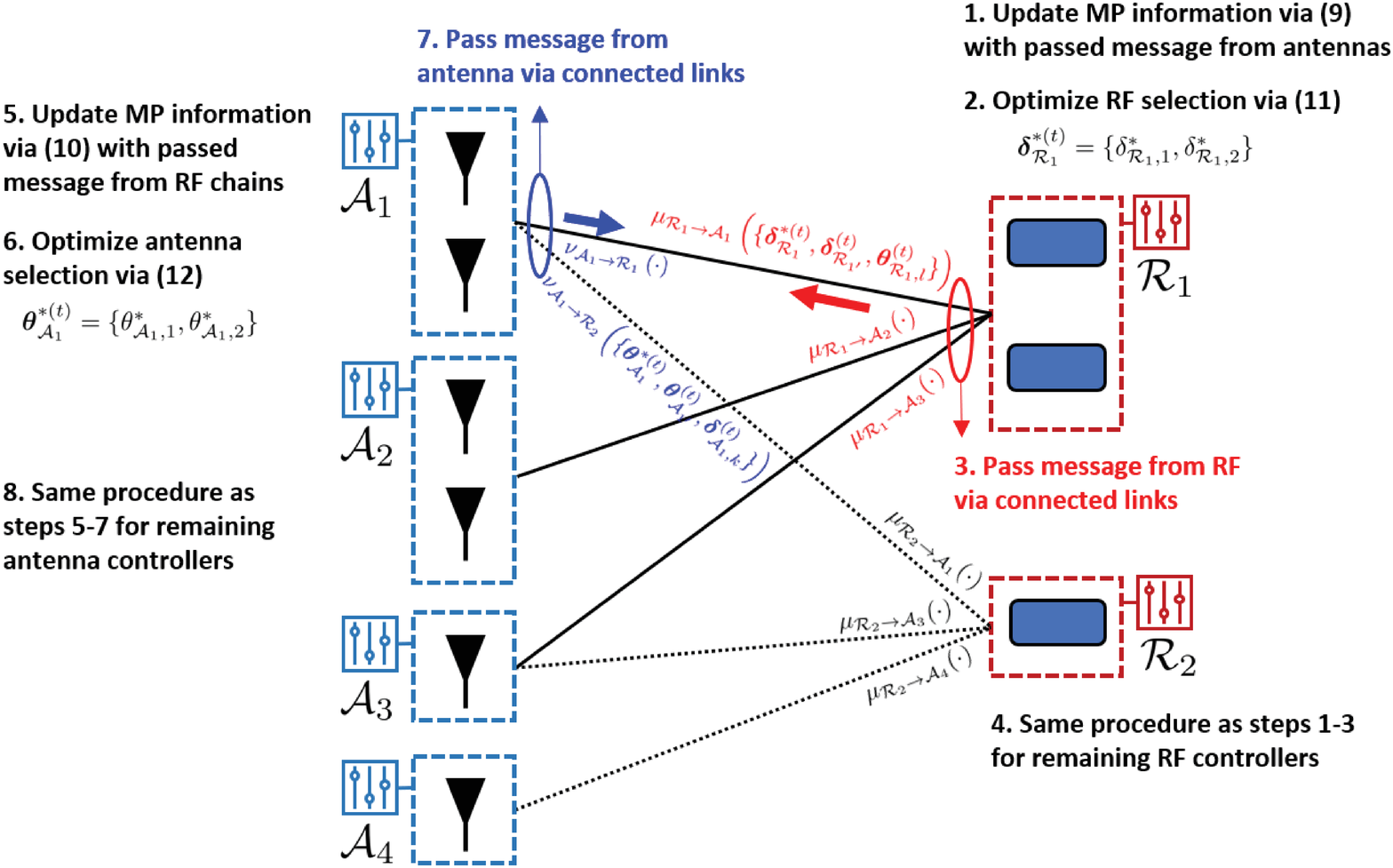}
    \caption{Example of sequential MP with the proposed MARS-S scheme considering $N_{ant}^C=4$ antenna controllers and $N_{RF}^C=2$ RF chain controllers.}\label{exp1}
\end{figure}

\subsubsection{Sequential Message Passing}
	After obtaining the optimal outcome of RF/antenna selection, the first controller passes its latest information to the neighboring connected nodes based on the LDPC-based HBF architecture. Then, the second controller managing its RF chains or antennas determines its own solution based on the information received from the neighboring controllers. That is, only one controller conducts MP while the others wait for the completion of message transfer. Therefore, the $k$-th RF controller employs the MP operation $\mu_{\mathcal{R}_{k} \rightarrow \mathcal{A}_{l}} \left( \{ \boldsymbol{\delta}_{\mathcal{R}_k}^{*(t)}, \boldsymbol{\delta}_{ {\mathcal{R}}_{k'} }^{(t)},  \boldsymbol{\theta}_{\mathcal{R}_{k},l}^{(t)}  \} \right)$ to pass its optimal RF selection $\boldsymbol{\delta}_{\mathcal{R}_k}^{*(t)}$ at the $t$-th iteration along with the updated information given in $\eqref{para_update_RF}$. Similarly, the $l$-th antenna controller transfers the optimal antenna selection and update information using $\nu_{\mathcal{A}_{l} \rightarrow \mathcal{R}_{k}} \left( \{ \boldsymbol{\theta}_{\mathcal{A}_l}^{*(t)}, \boldsymbol{\theta}_{ {\mathcal{A}}_{l'} }^{(t)}, \boldsymbol{\delta}_{\mathcal{A}_{l},k}^{(t)}\} \right)$. 
	 
The concrete scheme of the proposed MARS-S as a sequential MP is described in Algorithm \ref{alg:MARS-S}, which is also depicted in Fig. \ref{exp1}. All parameters are initialized randomly along with the measured channel information. First, the RF chain controller updates its decision according to the received message based on $\eqref{para_update_RF}$. We consider a randomized update method to prevent obtaining a locally optimal solution, i.e., $\eqref{optrf}$ is optimized when $x \leq \eta_{r}$, where $x$ is a random variable between $\left[0,1 \right]$ and $\eta_{r}$ is a predefined learning rate. After obtaining the temporary optimal RF selection $\boldsymbol{\delta}_{\mathcal{R}_{k}}^{*(t)}$ constrained by $\eqref{conrf}$, the $k$-th RF chain controller passes the decision to the neighboring antenna controllers and updates the set $\{ \boldsymbol{\delta}_{\mathcal{R}_k}^{*(t)}, \boldsymbol{\delta}_{ {\mathcal{R}}_{k'} }^{(t)},  \boldsymbol{\theta}_{\mathcal{R}_{k},l}^{(t)} \}$. The antenna controllers can perform their operations until the RF controllers have finished. Afterwards, the antenna controllers adopt a process similar to that of the RF chain controllers to update, optimize, and pass messages. The update of antennas is based on $\eqref{para_update_antenna}$, while the optimization follows $\eqref{optant}$ constrained by $\eqref{conant}$. The optimal solution can be acquired if the random value $x$ is smaller than the given learning rate $\eta_{a}$ for antenna selection. Then, the $l$-th antenna controller transfers the decision to the connected RF controllers and updates the set $\{ \boldsymbol{\theta}_{\mathcal{A}_l}^{*(t)}, \boldsymbol{\theta}_{ {\mathcal{A}}_{l'} }^{(t)}, \boldsymbol{\delta}_{\mathcal{A}_{l},k}^{(t)}\}$. Convergence occurs when the difference in power consumption between two iterations is smaller than a given threshold $\kappa$, i.e., $| P(\boldsymbol{\delta}^{(t-1)},\boldsymbol{\theta}^{(t-1)}) - P(\boldsymbol{\delta}^{(t)},\boldsymbol{\theta}^{(t)}) | \leq \kappa$.

\begin{algorithm}[!tb]
\scriptsize
\caption{Proposed MARS-S Scheme}
\SetAlgoLined
\DontPrintSemicolon
\label{alg:MARS-S}
\begin{algorithmic}[1]
\STATE Initialization: $\mathbf{W_{BB}}, \mathbf{W_{RF}}, \mathbf{H}, \mathbf{C}, \boldsymbol{\delta}, \boldsymbol{\theta}, \eta_{r}, \eta_{a}, \kappa, t=1$
\STATE Channel estimation is performed to obtain $\mathbf{H}$ 
\STATE The beamformers $\mathbf{W_{BB}}, \mathbf{W_{RF}}$ are randomized 
\STATE An LDPC-based connection is obtained based on Algorithm \ref{alg:LDPC}
\STATE The BS broadcasts the above information to each controller
\REPEAT
\STATE {\textbf{(RF Chain Controller)}} \\
\FOR{$\mathcal{R}_{k}\in \mathcal{R}^{\mathcal{C}}$}
\STATE  The $k$-th RF controller's information is updated based on $\eqref{para_update_RF}$
\STATE A number $x \in [0,1]$ is randomly generated
\IF{$x \leq \eta_{r}$}
\STATE The temporary optimized RF chain selection  $\boldsymbol{\delta}_{\mathcal{R}_{k}}^{*(t)}$ is obtained according to $\eqref{optrf}$
\STATE RF chains are reselected if $\eqref{conrf}$ is not satisfied
\ENDIF
\STATE The updated and optimized message is passed to neighboring connected controllers as $\mu_{\mathcal{R}_{k} \rightarrow \mathcal{A}_{l}} \left( \{ \boldsymbol{\delta}_{\mathcal{R}_k}^{*(t)}, \boldsymbol{\delta}_{ {\mathcal{R}}_{k'} }^{(t)},  \boldsymbol{\delta}_{\mathcal{A}_{l},k}^{(t)}  \} \right)$
\ENDFOR
\STATE {\textbf{(Antenna Controller)}} \\
\FOR{$\mathcal{A}_{l}\in \mathcal{A}^{\mathcal{C}}$}
\STATE The $l$-th antenna controller's information is updated according to $\eqref{para_update_antenna}$
\STATE A random number $x \in [0,1]$ is generated
\IF{$x \leq \eta_{a}$}
\STATE The temporary optimized antenna selection  $\boldsymbol{\theta}_{\mathcal{A}_{l}}^{*(t)}$ is derived according to $\eqref{optant}$
\STATE If $\eqref{conant}$ is not satisfied, antennas are reselected
\ENDIF
\STATE The updated and optimized message is passed to neighboring connected controllers as $\nu_{\mathcal{A}_{l} \rightarrow \mathcal{R}_{k}} \left( \{ \boldsymbol{\theta}_{\mathcal{A}_l}^{*(t)}, \boldsymbol{\theta}_{ {\mathcal{A}}_{l'} }^{(t)}, \boldsymbol{\delta}_{\mathcal{A}_{l},k}^{(t)}\} \right)$
\ENDFOR
\STATE {\textbf{(Optimum Derivation)}}
\STATE The optimum of RF/antenna selection is $\{ \boldsymbol{\delta}^{*},\boldsymbol{\theta}^{*} \}= \{\boldsymbol{\delta}^{(t)},\boldsymbol{\theta}^{(t)}\}$
\STATE Iteration update $t=t+1$
\UNTIL{Convergence of $| P(\boldsymbol{\delta}^{(t-1)},\boldsymbol{\theta}^{(t-1)}) - P(\boldsymbol{\delta}^{(t)},\boldsymbol{\theta}^{(t)}) | \leq \kappa$}
\end{algorithmic}
\end{algorithm}

\subsection{MARS-P for Parallel Message Passing}

\begin{figure}
    \centering
\includegraphics[width=3.3 in]{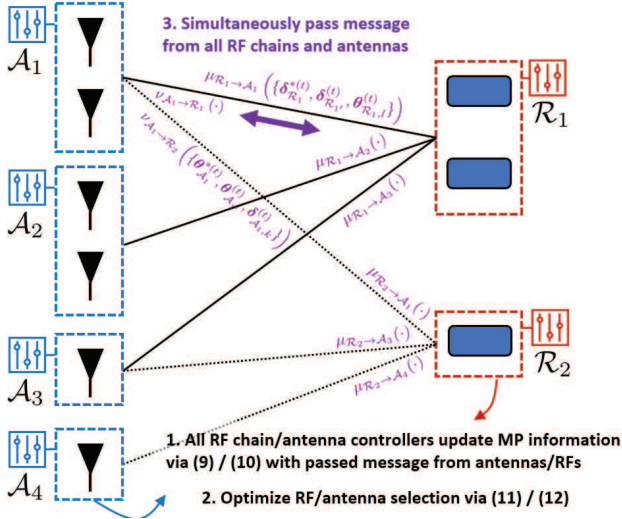}
    \caption{Example of parallel MP with the proposed MARS-P scheme considering $N_{ant}^C=4$ antenna controllers and $N_{RF}^C=2$ RF chain controllers.}\label{exp2}
\end{figure}

\begin{algorithm}[!tb]
\scriptsize
\caption{Proposed MARS-P Scheme}
\SetAlgoLined
\DontPrintSemicolon
\label{alg:MARS-P}
\begin{algorithmic}[1]
\STATE Initialization: $\mathbf{W_{BB}}, \mathbf{W_{RF}}, \mathbf{H}, \mathbf{C}, \boldsymbol{\delta}, \boldsymbol{\theta}, \eta_{r}, \eta_{a}, \kappa, t=1$
\STATE Channel estimation is performed to obtain $\mathbf{H}$
\STATE The beamformers $\mathbf{W_{BB}}, \mathbf{W_{RF}}$ are randomized
\STATE An LDPC-based connection is obtained based on Algorithm \ref{alg:LDPC}
\STATE The BS broadcasts the above information to each controller
\REPEAT
\STATE {\textbf{(RF Chain Controller)}} \\
\STATE  All RF controllers $\mathcal{R}_{k} \in \mathcal{R}^{\mathcal{C}}$ update their information according to $\eqref{para_update_RF}$
\STATE A temporary optimized RF chain selection $\boldsymbol{\delta}_{\mathcal{R}_{k}}^{*(t)}$ is acquired according to $\eqref{optrf}$ when the generated number of $x\in [0,1]$ is smaller than $\eta_{r}$
\STATE RF chains are reselected if $\eqref{conrf}$ is not satisfied
\STATE {\textbf{(Antenna Controller)}} \\
\STATE All antenna controllers $\mathcal{A}_{l} \in \mathcal{A}^{\mathcal{C}}$ update their information based on $\eqref{para_update_antenna}$
\STATE A temporary optimized antenna selection result is acquired $\boldsymbol{\theta}_{\mathcal{A}_{l}}^{*}$ in $\eqref{optant}$ when the generated number $x\in [0,1]$ is smaller than $\eta_{a}$
\STATE Antennas are reselected if $\eqref{conant}$ is not satisfied
\STATE {\textbf{(Simultaneous Message Passing)}} \\
\STATE All RF/antenna controllers simultaneously transfer the decision and latest message to the connected antenna/RF controllers by MP: $\mu_{\mathcal{R}_{k} \rightarrow \mathcal{A}_{l}} \left( \{ \boldsymbol{\delta}_{\mathcal{R}_k}^{*(t)}, \boldsymbol{\delta}_{ {\mathcal{R}}_{k'} }^{(t)},  \boldsymbol{\delta}_{\mathcal{A}_{l},k}^{(t)}  \} \right)$ and
$\nu_{\mathcal{A}_{l} \rightarrow \mathcal{R}_{k}} \left( \{ \boldsymbol{\theta}_{\mathcal{A}_l}^{*(t)}, \boldsymbol{\theta}_{ {\mathcal{A}}_{l'} }^{(t)}, \boldsymbol{\delta}_{\mathcal{A}_{l},k}^{(t)}\} \right), \forall \mathcal{R}_{k} \in \mathcal{R}^{\mathcal{C}}, \forall \mathcal{A}_{l} \in \mathcal{A}^{\mathcal{C}}$
\STATE {\textbf{(Optimum Derivation)}}
\STATE The optimum of RF/antenna selection is $\{ \boldsymbol{\delta}^{*},\boldsymbol{\theta}^{*} \}= \{\boldsymbol{\delta}^{(t)},\boldsymbol{\theta}^{(t)}\}$
\STATE Iteration update $t=t+1$
\UNTIL{Convergence of $| P(\boldsymbol{\delta}^{(t-1)},\boldsymbol{\theta}^{(t-1)}) - P(\boldsymbol{\delta}^{(t)},\boldsymbol{\theta}^{(t)}) | \leq \kappa$}
\end{algorithmic}
\end{algorithm}

	We can infer from MARS-S that it potentially induces a low convergence speed due to the sequential operation of MP. The adjacent controllers should wait until the completion of their connected controllers. Therefore, we design a parallel MP type, MARS-P, so that all RF/antenna controllers can simultaneously pass their optimized determinations. The steps of the proposed MARS-P scheme are similar to those of MARS-S for initialization, updating of $\eqref{para_update_RF}$ and $\eqref{para_update_antenna}$, and optimization of $\eqref{optrf}$ and $\eqref{optant}$ for RF and antenna selection, respectively. The only difference between MARS-P and MARS-S is the way in which parallel MP conveys information. The detailed algorithm of MARS-P is presented as Algorithm \ref{alg:MARS-P}, which enables all controllers to simultaneously transmit messages without waiting. An example to demonstrate the process of MARS-P is illustrated in Fig. \ref{exp2}. However, due to the simultaneous update in MARS-P, each controller potentially suffers from local decisions, leading to a low convergence rate. Therefore, similar to MARS-S, we also need to design a feasible learning rate parameter to strike a compelling tradeoff between convergence and local optimality.
	
	To elaborate a little further, the proposed MARS scheme can be used in both narrowband and wideband systems. There are two ways to modify the MARS algorithms for these systems. The first method is through an \textbf{averaging} approach, which involves estimating the wideband channel state information using a single parameter with the dimension of the multiplied numbers of transmitters and receivers. However, this approach results in lower complexity but worse performance due to the coarse channel estimation compared to fine-grained subchannel optimization. The second scheme is from a \textbf{subchannel}-based perspective, which provides better performance but is more complicated due to the coupled total power among subchannels. This may result in a trade-off among different subchannels, where selecting a certain subchannel may be harmful to another one. To address this issue, an alternative optimization approach can be designed by adding an additional iteration loop to iteratively search over each subchannel until convergence. This approach enables us to apply the existing MARS algorithm directly in a wideband scenario, but it may require additional computational complexity depending on the number of subchannels being operated. Moreover, advanced beamforming can be designed to further enhance the system performance, whereas this work can be regarded as a lower bound of joint optimization of both RF/antenna selection and beamformer design.

\subsection{Heuristic Hybrid Beamformer Solution}

After obtaining the selection policy, we proceed to conduct hybrid beamformer scheme which is designed based on continuous genetic algorithm \cite{cga}. Genetic-based hybrid beamforming algorithm is composed of a series of genetic process, including gene generation initialization, elite selection, crossover, and mutation. The detailed process is elaborated as follows. \textbf{Initial Gene Generation}: Initially, we randomly generate a gene set $\mathcal{X} = \left\lbrace \mathcal{X}_{1},...,\mathcal{X}_{g},..., \mathcal{X}_{N_{G}} \right\rbrace$, where $N_{G}$ is the population of the genes. Each gene $\mathcal{X}_{g}$ comprises a concatenating vector of candidate HBF solutions $\{ \mathbf{W}_{\mathbf{RF}}, \mathbf{W}_{\mathbf{BB}} \}$ given by $\mathcal{X}_{g} \in \mathbb{R}^{2(N_{RF} N_{ant} + N_S N_{RF})}$
\begin{align} \label{ggg}
	\mathcal{X}_{g} &= \left[ 
	\left(\text{Flat}(\mathfrak{R}\{ {\mathbf{W}_{\mathbf{RF}}} \}) \right)^T,
	\left(\text{Flat}(\mathfrak{I}\{ {\mathbf{W}_{\mathbf{RF}}} \}) \right)^T, \right. \notag \\	
	& \left. \qquad \left(\text{Flat}(\mathfrak{R}\{ {\mathbf{W}_{\mathbf{BB}}} \}) \right)^T,
	\left(\text{Flat}(\mathfrak{I}\{ {\mathbf{W}_{\mathbf{BB}}} \}) \right)^T
\right],
\end{align}	
where $\text{Flat}(\cdot)$ vectorizes a matrix into a column vector. We further notice that we generate $\vartheta \in [0, 2\pi)$ in $e^{j \vartheta}$ in order to satisfy the unit modulus equality constraint in \eqref{con_RF} for $\mathbf{W}_{\mathbf{RF}}$.

After obtaining the gene solution set, we proceed to conduct \textbf{Elite Gene Selection}. In elite selection, the fitness value will be evaluated in order to evaluate the effectiveness of gene solution. It can be observed from problem \eqref{mainprob} that we only need to meet the requirement of constraint \eqref{SR_req}. Therefore, the fitness function is designed as
\begin{align} \label{fitness}
	F (\mathcal{X}_{g}) \!=\! 
\begin{cases} 
R(\mathbf{\Delta^*,\Theta^*}, \mathcal{X}_{g}) \!-\! R_{req},  & \mbox{if } R(\mathbf{\Delta^*,\Theta^*}, \mathcal{X}_{g}) \!\geq\! R_{req}, \\
-\xi, & \mbox{otherwise,}
\end{cases}
\end{align}
where $\mathbf{\Delta^*,\Theta^*}$ are the RF/antenna selection solution obtained in MARS schemes in either Algorithm \ref{alg:MARS-S} or \ref{alg:MARS-P}, respectively. If not satisfying the constraint, it provides a penalty value with a sufficient large negative constant $-\xi$. Therefore, we select the top-$N_{E}$ genes with higher fitness values as the elite group, whereas the remaining unqualified genes are abandoned. During \textbf{Crossover}, we will generate new offspring genes from the elite group. We firstly randomly select two elite genes $\mathcal{X}_{g}$ and $\mathcal{X}_{g'}, \forall g'\neq g$. Then, we generate a random continuous sequence $\mathcal{S}$ with each element having a range of $[0,1]$. Note that the length of $\mathcal{S}$ is the same as that of $\mathcal{X}_{g}$. Accordingly, we can generate two offspring genes $\mathcal{X}_{\widetilde{g}}$ and $\mathcal{X}_{\widetilde{g}'}$ respectively as
\begin{align} \label{cross}
	\mathcal{X}_{\widetilde{g}} &= (\mathcal{X}_{g} \circ \mathcal{S}) + \Big( \mathcal{X}_{g'} \circ (1- \mathcal{S}) \Big), \\
	\mathcal{X}_{\widetilde{g}'} &= \Big( \mathcal{X}_{g} \circ (1-\mathcal{S}) \Big) + (\mathcal{X}_{g'} \circ  \mathcal{S}).
\end{align}	
A total of $N^B_{crx}$ new offspring genes will be generated from all possible combinations. In order to prevent local optimum solution, we proceed to perform \textbf{Mutation} operation. The concept of mutation is to alternate some elements of a gene. That is, a total of $N_{mu}^{B}$ elements will be selected for mutation. The values of those selected gene elements will be regenerated with the corresponding feasible ranges, i.e., the value regarding baseband beamforming is within a range of $\left[ \min\, \mathcal{X}_g(\mathbf{W_{BB}}) , \max \, \mathcal{X}_g(\mathbf{W_{BB}}) \right]$, whilst that regarding analog beamforming follows the range of $\left[ \min\, \mathcal{X}_g(\mathbf{W_{RF}}) , \max \, \mathcal{X}_g(\mathbf{W_{RF}}) \right]$. As for the next genetic generation, procedures from the steps of elite selection, crossover and mutation will be repeatedly performed until the termination criteria are satisfied, i.e., $\max \, F^{(\tau)}(\mathcal{X}_g) \leq \iota_{1}$ and $\left| \max \, F^{(\tau)}(\mathcal{X}_g) - \max \, F^{(\tau-1)}(\mathcal{X}_g)  \right| \leq \iota_{2}$. The first inequality guarantees the optimal fitness value lower than $\iota_1$. While, the second inequality indicates the the difference between the fitness values at the $\tau$-th and at the $(\tau-1)$-th generation should be sufficiently small. The overall algorithm is elaborated in Algorithm \ref{overall}.

\begin{algorithm}[!tb]
\scriptsize
\caption{Joint RF/Antenna Selection and HBF Solution}
\SetAlgoLined
\DontPrintSemicolon
\label{overall}
\begin{algorithmic}[1]
\STATE Initialization: $\mathbf{W_{BB}}, \mathbf{W_{RF}}, \boldsymbol{\delta}, \boldsymbol{\theta}, \mathbf{H}, \mathbf{C}$
\REPEAT
	\STATE Perform LDPC-based connection in Algorithm \ref{alg:LDPC} for $\mathbf{C}$
	\STATE Conduct RF/antenna selection scheme in Algorithm \ref{alg:MARS-S} if sequential MP or in Algorithm \ref{alg:MARS-P} if parallel MP
	\STATE Obtain optimal selection solution $\boldsymbol{\delta}^*, \boldsymbol{\theta}^*$
	\STATE Perform genetic-algorithm based hybrid beamforming solution, with initial gene generation based on \eqref{ggg} 
	\REPEAT
		\STATE Execute elite gene selection based on sorting fitness values of \eqref{fitness}
		\STATE Conduct crossover operation based on \eqref{cross}
		\STATE Perform mutation operation with  $N_{mu}^{B}$ elements with their corresponding ranges
	\UNTIL{Convergence satisfying $\max \, F^{(\tau)}(\mathcal{X}_g) \leq \iota_{1}$ and $\left| \max \, F^{(\tau)}(\mathcal{X}_g) - \max \, F^{(\tau-1)}(\mathcal{X}_g)  \right| \leq \iota_{2}$}
\UNTIL	
\end{algorithmic}
\end{algorithm}

\subsection{Complexity Analysis}
	The computational complexity is demonstrated in Table \ref{complex}. We know that the joint solution for the conventional full-connection HBF architecture in problem $\eqref{mainprob}$ requires an exponential complexity of $\mathcal{O} \left(2^{N_{RF}N_{ant}}\right)$, which leads to a potential difficulty in acquiring the global optimum. On the other hand, the proposed MARS scheme possesses a much lower computational complexity than the exhaustive search of the original problem, i.e., MARS-S as a sequential MP scheme achieves a complexity of $\mathcal{O}\left( 2^{|\mathcal{R}^{\mathcal{C}} | N_{RF,k} + |\mathcal{A}^{\mathcal{C}} | N_{ant,l} } \right)$, whereas the parallel-type mechanism of MARS-P possesses a complexity of $\mathcal{O}\left( 2^{N_{RF,k}+ N_{ant,l} } \right)$. It is observed that MARS-S has a higher computational complexity than MARS-P due to the sequential MP mechanism. The RF/antenna controllers update their latest information and then determine the corresponding selection results, which leads to no missed messages. In contrast, greedy-based selection \cite{49} only considers a single-antenna policy with others fixed, which has a complexity order of $\mathcal{O} \left(|\mathcal{R}^{\mathcal{C}} | N_{RF,k}+ | \mathcal{A}^{\mathcal{C}} | N_{ant,l} \right)$. Moreover, the genetic-based selection method \cite{50}, which adopts genetic generation, elite selection, crossover and mutation, possesses a complexity order of $\mathcal{O}\left( N_{crx}N_{RF}N_{ant} + N_{mu} \right)$, where $N_{crx}$ and $N_{mu}$ denote the numbers of crossover and mutation operations for selection, respectively. Since both papers \cite{49,50} only consider antenna selection, we therefore consider the genetic and greedy selection to the RF selection problem for fair comparison. Due to the simultaneous passing of information in MARS-P, there may be either missed or out-of-date information, potentially resulting in a locally optimal solution. However, compared to the full connection of the HBF scheme in the original problem, the proposed MARS scheme reaches the lowest computational complexity while realizing nearly optimal solution acquisition. Additionally, optimizing genetic-based hybrid beamforming algorithm requires a complexity order of $\mathcal{O}\left(  N_{crx}^{B}N_{RF}^2 N_{ant} N_{S} + N_{mu}^{B} \right)$, where $N_{crx}^{B}$ and $N_{mu}^{B}$ indicate the numbers of crossover and mutation operation for hybrid beamforming, respectively.

\begin{table}
\begin{center}
\small
\caption {Computational Complexity}
    \begin{tabular}{|l|r|}
    \hline
        Algorithm & Complexity \\
        \hline\hline
		Global Optimum of Selection & $\mathcal{O} \left(2^{N_{RF}N_{ant}}\right)$ \\ \hline
		Proposed MARS-S & $\mathcal{O}\left( 2^{|\mathcal{R}^{\mathcal{C}} | N_{RF,k}+ | \mathcal{A}^{\mathcal{C}} | N_{ant,l} } \right)$ \\ \hline
		Proposed MARS-P & $\mathcal{O}\left( 2^{N_{RF,k}+ N_{ant,l} } \right)$  \\ \hline
		Greedy-Based Selection & $\mathcal{O} \left( |\mathcal{R}^{\mathcal{C}} | N_{RF,k}+ | \mathcal{A}^{\mathcal{C}} | N_{ant,l}  \right)$ \\ \hline
		Genetic-Based Selection & $\mathcal{O}\left(  N_{crx}N_{RF}N_{ant} + N_{mu} \right)$\\ \hline
		Genetic-Based Beamforming & $\mathcal{O}\left(  N_{crx}^{B}N_{RF}^2 N_{ant} N_{S} + N_{mu}^{B} \right)$\\ \hline
    \end{tabular} \label{complex}
\end{center}
\end{table}


\begin{table}
\begin{center}
\footnotesize
\caption {System Parameter Settings}
    \begin{tabular}{|l|l|}
    \hline
        Parameter & Value \\        
        \hline\hline
        Carrier frequency $f_c$ & $28$ GHz \\ \hline
        Distance between transmitter and receiver $d$ & $100$ meters \\ \hline
		Number of data stream $N_{S}$ & $4$ \\ \hline
		Transmit power $P_{T}$ of user & $20$ dBm \\ \hline
		Noise power $N_{0}$ & $-85$ dBm \\ \hline
		BS power consumption of an RF chain $P_{RF}$ & $40$ mWatt \\ \hline
		BS power consumption of baseband processing $P_{BB}$ & $800$ mWatt \\ \hline
		BS power consumption of an ADC $P_{ADC}$ & $100$ mWatt \\ \hline
		BS power consumption of a phase shifter $P_{PS}$ & $10$ mWatt \\ \hline
		BS power consumption of an LNA $P_{LNA}$ & $10$ mWatt \\ \hline
		Maximum power support of each RF chain $P_{o}$ & $25$ dBm \\ \hline
		Maximum allowable BS operating power $P_{max}$ & $44$ dBm \\ \hline
        Thresholds for RF and antenna selection $\eta_{r},\eta_{a}$ & $0.7, 0.7$ \\ \hline
        Population of genes $N_G$ & $100$ \\ \hline
        Number of crossover genes $N_{crx}^B$ & $100$ \\ \hline
        Number of mutation elements $N_{mu}^B$ & $0.1\cdot N_G N_{crx}^B$ \\ \hline
        Penalty term $\xi$ & $10^3$ \\ \hline
        Convergence thresholds $\iota_1$, $\iota_2$ & $10^{-1}$, $10^{-1}$ \\ \hline
    \end{tabular} \label{Parameter}
\end{center}
\end{table}

\section{Performance Evaluation} 

The performance of the proposed MARS scheme is evaluated via simulations. We consider an uplink transmission with a single transmitter and a hybrid beamformed-MIMO receiver. The distance between the transmitter and receiver is set to $d=100$ meters, and the operating carrier frequency is $f_c=28$ GHz. The transmit power of the user is set to $20$ dBm \cite{uepower}. The receiver power consumption values of the BS utilized in our simulation are $P_{RF}=40$ mWatt, $P_{BB}=800$ mWatt, $P_{ADC}=100$ mWatt, and $P_{PS}=P_{LNA}=10$ mWatt \cite{12,bb,lna,pslna}. The maximum power support for each RF chain associated with the maximum number of the connected receiving antennas is $P_{o}=25$ dBm, and the allowable system operating power at receiving BS is $P_{max}=44$ dBm. The parameter $\rho=10^{-PL/10}$ refers to the distance-based path loss, where $PL=32.4+20\log_{10}(f_{c})+30\log_{10}(d)$ \cite{spec}. The learning rate threshold for RF/antenna selection is $\eta_{r}=\eta_{a}=0.7$. For simulation simplicity, we consider equivalent numbers of elements in each RF/antenna group, i.e., $N_{RF,k}=N_{RF}^{G}$ and $N_{ant,l}=N_{ant}^{G}, \forall k,l$. The related outcome is characterized by power consumption (W) and corresponding EE (bits/J/Hz)\textsuperscript{\ref{note5}}\footnotetext[5]{The letter 'W' denotes "Watts", while 'J' denotes "Joules", which is the product of time in seconds (s) and power in Watts.\label{note5}}, i.e., $EE=R(\mathbf{\Delta,\Theta}, \mathbf{W}_{\mathbf{RF}}, \mathbf{W}_{\mathbf{BB}})/P(\boldsymbol{\delta},\boldsymbol{\theta})$. The pertinent system parameters are listed in Table \ref{Parameter}. We further notice that the proposed MARS scheme includes both RF/antenna selection and HBF in the following simulation results.

\subsection{Convergence and Learning Rate}

\begin{figure}[!t]
\centering
\subfigure[]{\includegraphics[width=1.65 in]{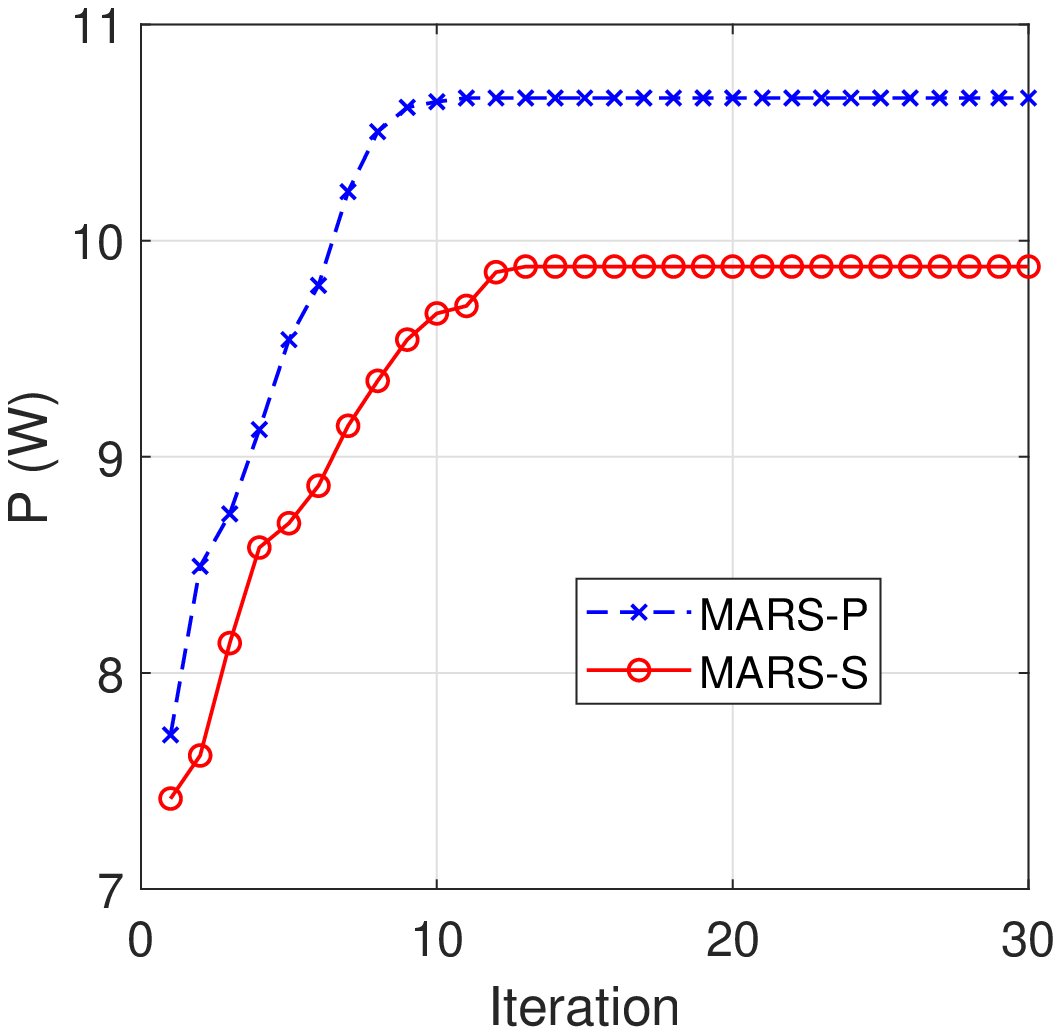} \label{iter1}}
\subfigure[]{\includegraphics[width=1.65 in]{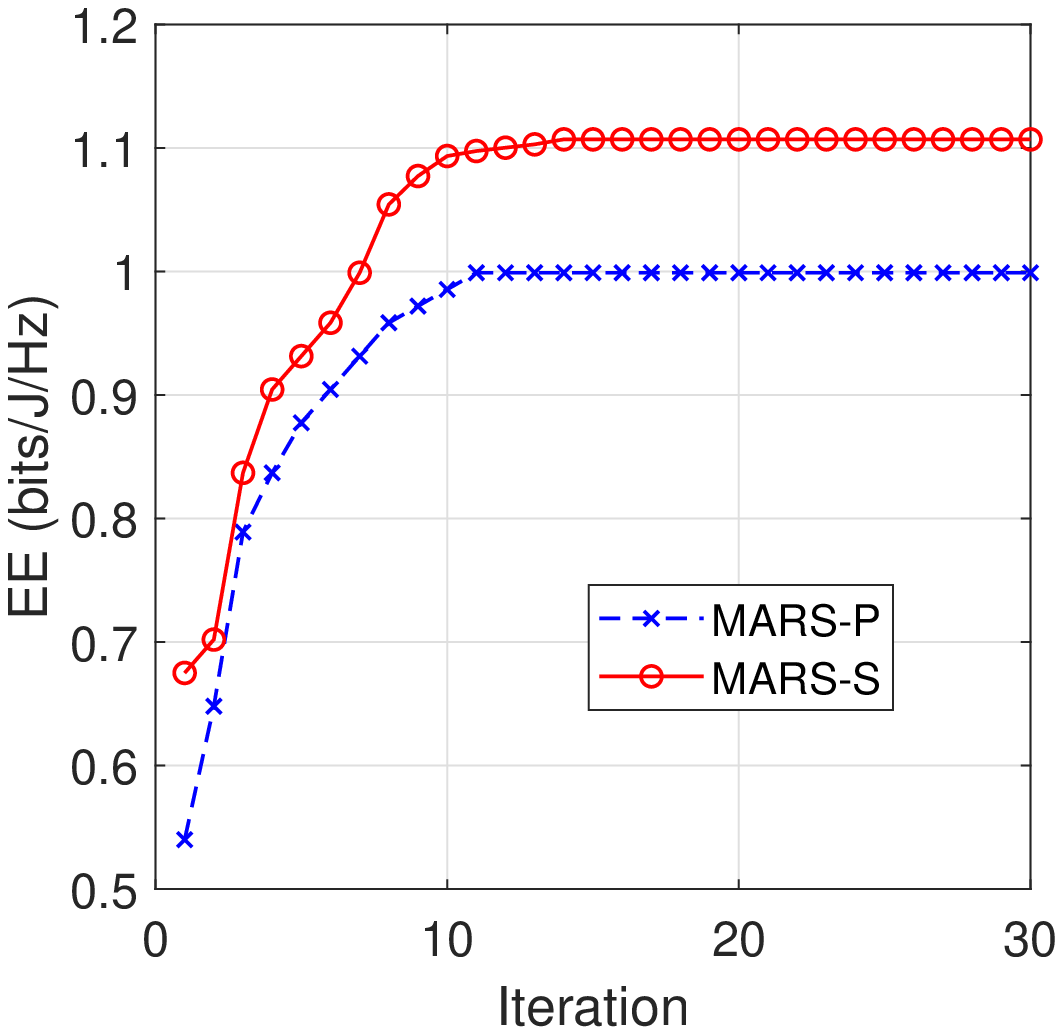} \label{iter3}}
\caption{Convergence of the proposed MARS-S and MARS-P subschemes in terms of (a) power consumption and (b) EE.} \label{iter}
\end{figure}

As illustrated in Fig. \ref{iter}, the convergence of the proposed MARS scheme is validated through simulations. We infer from both figures that the proposed MARS-S achieves lower power consumption and higher EE than MARS-P. This is because with sequential passing, up-to-date decisions can be completely conveyed to other nodes in order, whereas the parallel scheme potentially confuses RF/antenna controllers with multiple simultaneous received information. This results in worse information updates and corresponding low-performance policies. Moreover, because a better policy can be obtained from MARS-S, it requires more iterations to reach convergence. That is, MARS-S needs approximately $17$ iterations resulting in the power consumption of approximately $9.88$ W and EE of $1.1$ bits/J/Hz, while MARS-P converges with higher power consumption and lower EE at the $10$-th round, which strikes a compelling tradeoff between convergence speed and performance.

\begin{figure}[!t]
\centering
\subfigure[]{\includegraphics[width=\m in]{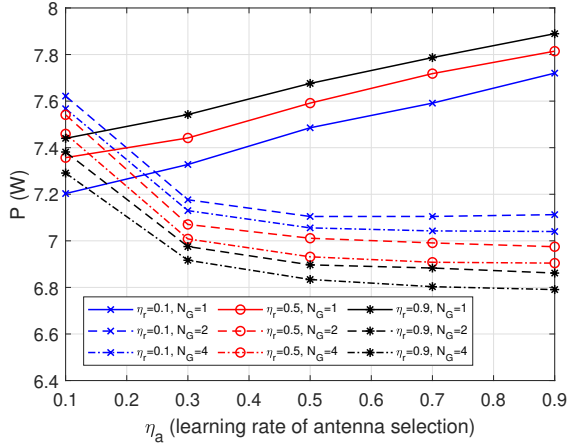} \label{lr1}}
\subfigure[]{\includegraphics[width=\m in]{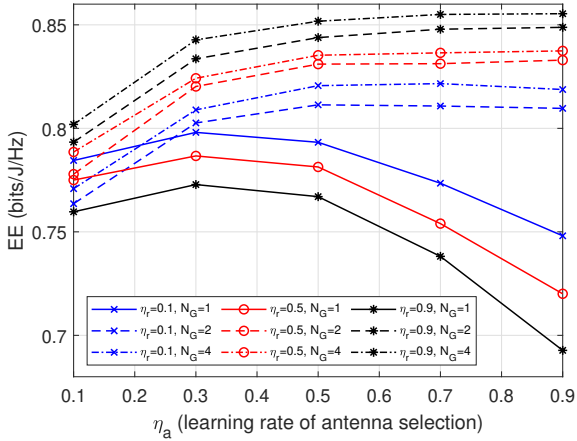} \label{lr3}}
\caption{Performance of the proposed MARS scheme considering different learning rates for antenna selection $\eta_{a} \in \{0.1,0.3,0.5,0.7,0.9\}$ and RF selection $\eta_{r}\in\{0.1,0.5,0.9\}$ with $N_{G} = N_{RF}^{G} = N_{ant}^{F} \in \{1,2,4\}$ elements in each RF/antenna group with respect to (a) power consumption and (b) EE.} \label{lr}
\end{figure}

In Fig. \ref{lr}, we evaluate the performance in terms of power consumption and EE of the proposed MARS scheme under different learning rates $\eta_{a}\in\{0.1,0.3,0.5,0.7,0.9\}$ for antenna selection and $\eta_{r}\in\{0.1,0.5,0.9\}$ for RF chain selection. We consider the same number of RF/antenna elements per group, $N^{G}=N_{RF}^{G}=N_{ant}^{G}\in\{1,4\}$. We observe from Figs. \ref{lr1} and \ref{lr3} that higher learning rates result in lower power consumption and higher EE when more nodes are managed by RF/antenna controllers, i.e., $N_G\in\{2,4\}$. This is because a quicker policy update can prevent out-of-date information from being conveyed to other controllers under a more centralized architecture, which requires fewer hops to convey the decision policy. However, under a more distributed architecture with fewer nodes per group, namely, $N_{G}=1$, a faster rate of $\eta_{a}=0.9$ consumes more power, up to an average power of $7.89$ W, since the latest optimal message cannot be transferred to faraway nodes before the next update. The best policy is potentially obscured by newly passed messages from neighboring controllers. Moreover, concave EE curves are obtained for $N_{G}=1$ due to the moderate update speed, where the optimum is attained when $\eta_{a}=0.3$, as shown in Fig. \ref{lr3}.

\subsection{Different RF/Antenna Configurations}

\begin{figure}[!t]
\centering
\subfigure[]{\includegraphics[width=\m in]{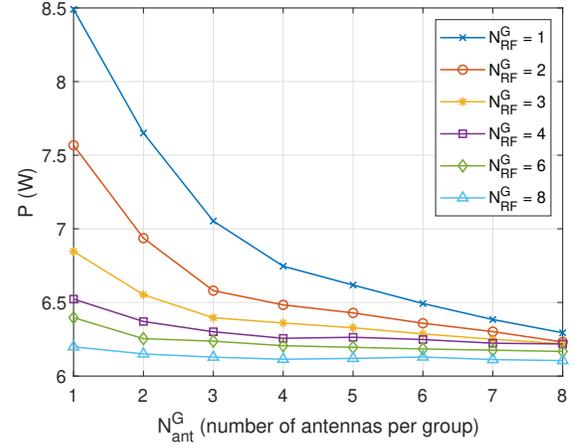} \label{ng1}}
\subfigure[]{\includegraphics[width=\m in]{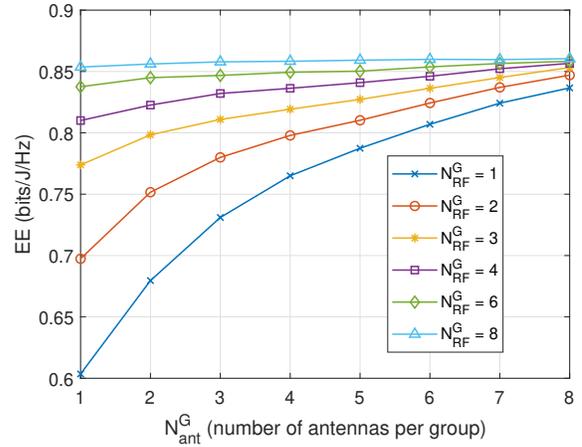} \label{ng3}}
\caption{Performance of MARS with various numbers of clustered RF/antenna elements in groups $N_{RF}^{G}\in\{1,2,3,4,6,8\}$ and $N_{ant}^{G}\in\{1,2,3,4,5,6,7,8\}$, respectively, in terms of (a) power consumption and (b) EE.} \label{ng}
\end{figure}

As shown in Fig. \ref{ng}, we evaluate the effects of different numbers of RF/antenna elements in each clustered group. We observe from Fig. \ref{ng1} that a more centralized architecture with more antennas per group with larger $N_{ant}^{G}$ achieves lower power consumption when $N_{ant}^{G}\in\{1,2,3\}$. The first reason for this phenomenon is that the controller is capable of obtaining a nearly optimal solution with a higher degree of freedom in selecting an on-off policy. The other reason is the shorter paths for conveying the determined message from one group to the others, which reveals similar effects, as illustrated in Fig. \ref{lr}. Furthermore, the curves for $N_{RF}^{G}\in\{6,8\}$ are flatter than those of $N_{ant}^{G}\in\{1,2,3,4\}$ since the selection of RF chains, with considerably higher power consumption, becomes more dominant than antenna selection. This also implies that fewer RF controllers are potentially able to provide a better policy with lower power consumption and accordingly higher EE performance.

\begin{figure}[!t]
\centering
\subfigure[]{\includegraphics[width=\m in]{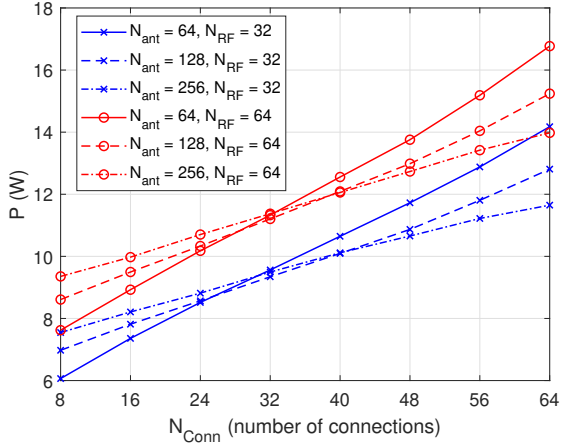} \label{Nc1}}
\subfigure[]{\includegraphics[width=\m in]{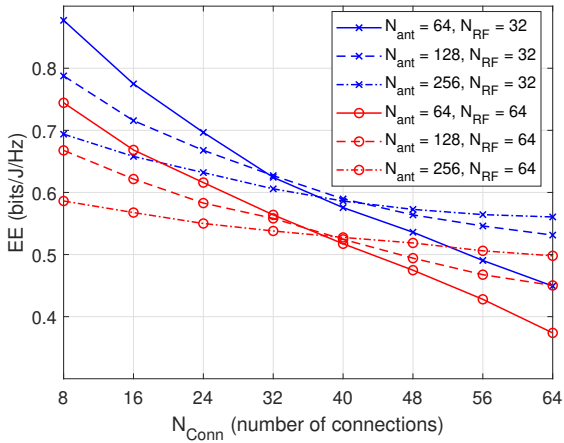} \label{Nc3}}
\caption{Effect of different numbers of connections $N_{conn}\in\{8,16,24,32,40,48,56,64\}$ with $N_{RF}\in\{32,64\}$ RF chains and $N_{ant}\in\{64,128,256\}$ antennas in terms of (a) power consumption and (b) EE.} \label{Nc}
\end{figure}

In Fig. \ref{Nc}, we show the impact of the proposed MARS algorithm under different numbers of connections among RF/antenna controllers and nodes from $N_{Conn}=8$ to $64$ with $N_{RF}\in\{32,64\}$ RF chains and $N_{ant}\in\{64,128,256\}$ antennas. We recall that $N_{Conn}$ is the number of antennas associated with a single RF chain, which implies a tendency to become a fully connected beamformer. When $N_{Conn}=64$, it has the highest power consumption and lowest EE due to simultaneous message passing from nodes in different groups, which leads to more uncertain and complex decision-making. For example, an RF controller will update various distinct messages delivered from massive antenna groups connected to it, which may result in a more inappropriate update and corresponding local solution. Moreover, more antennas result in less power consumption due to a higher degree of freedom of selection, e.g., power reduces from $16.7$ W to $13.9$ W when $N_{ant}$ increases from $64$ to $256$ with $N_{RF}=64$ RF chains and $N_{Conn} = 64$. With the aid of HBF, more antennas can be switched off for conserving power. However, RF selection exhibits a  reversed trend because it has higher operating power and lower selection freedom than antennas, which consumes approximately additional power of $2.6$ W when $N_{RF}=64$ compared to the case of $N_{RF}=32$ under $N_{ant}=64$ antennas and $N_{Conn} = 64$. To elaborate slightly further, we observe that another phenomenon takes place when we have the fewest connections, namely, $N_{Conn}\in\{8,16,24\}$. A higher number of antennas with fewer links can cause some optimal messages to be missed in the current iteration or confused with irrelevant information, deteriorating system performance.

\begin{figure}[!t]
\centering
\subfigure[]{\includegraphics[width=\m in]{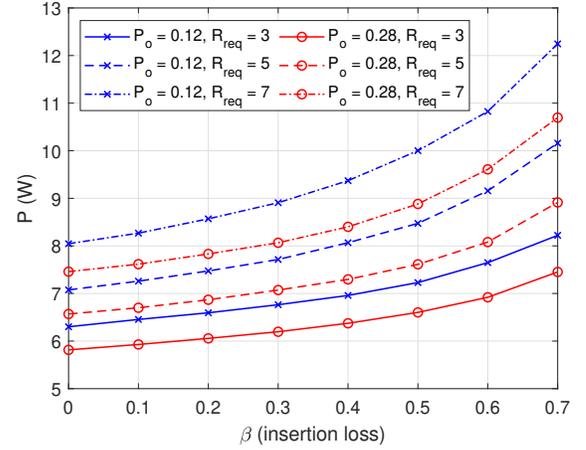} \label{p1}}
\subfigure[]{\includegraphics[width=\m in]{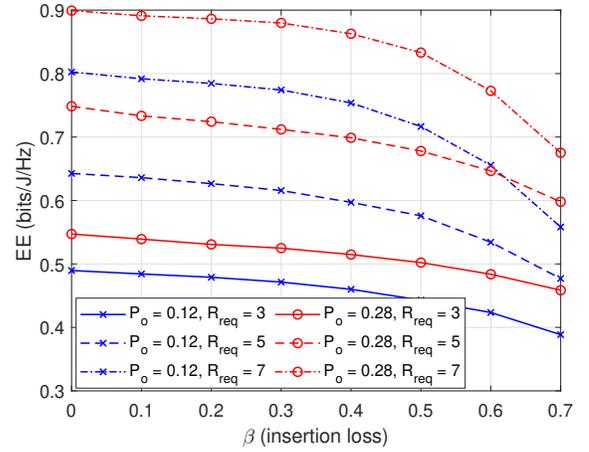} \label{p3}}
\caption{Performance of the proposed MARS scheme under different values of the insertion loss $\beta \in \{0,0.1,0.2,0.3,0.4,0.5,0.6,0.7\}$, maximum power support for each RF chain $P_{o}\in\{0.12,0.28\}$ W, and QoS $R_{req}\in\{3,5,7\}$ bits/s/Hz in terms of (a) power consumption and (b) EE.} \label{plim}
\end{figure}

	As depicted in Fig. \ref{plim}, the performance of the proposed MARS algorithm is evaluated considering different values of the insertion loss and maximum power support for each RF chain as well as QoS constraints. We observe from both the received signal in $\eqref{recvsignal}$ and throughput in $\eqref{recvcapacity}$ that both are monotonically decreasing and concave with respect to $\beta$. Thus, more power is consumed to compensate for the increased insertion loss, which exhibits a monotonically increasing convex shape, as shown in Fig. \ref{p1}. Moreover, higher capability in terms of the power support for each RF chain, i.e., $P_{o}=0.28$ W, requires lower power consumption due to the attainable higher degree of freedom in RF/antenna selection. Similarly, without any insertion loss, namely, when $\beta=0$, little difference in power consumption is observed due to the relaxed constraint in $\eqref{RF_insertionloss}$. Accordingly, most of the power can be utilized mainly for QoS satisfaction, not for compensation of insertion losses. With the increased rate demand, more power is intuitively required with a higher QoS requirement, e.g., we need approximately $3.2$ W more power in most cases when QoS increases from $3$ to $7$ bits/s/Hz. Moreover, as illustrated in Fig. \ref{p3}, decreasing EE curves would intersect with $\beta$ larger than $0.6$ and QoS of $\{5,7\}$ bits/s/Hz because the system requires considerably higher power for QoS $R_{req}=7$ bits/s/Hz to simultaneously compensate for substantial insertion losses and satisfy the QoS requirement.

\subsection{Benchmarking}

\begin{figure}
\centering
\subfigure[]{\includegraphics[width=\m in]{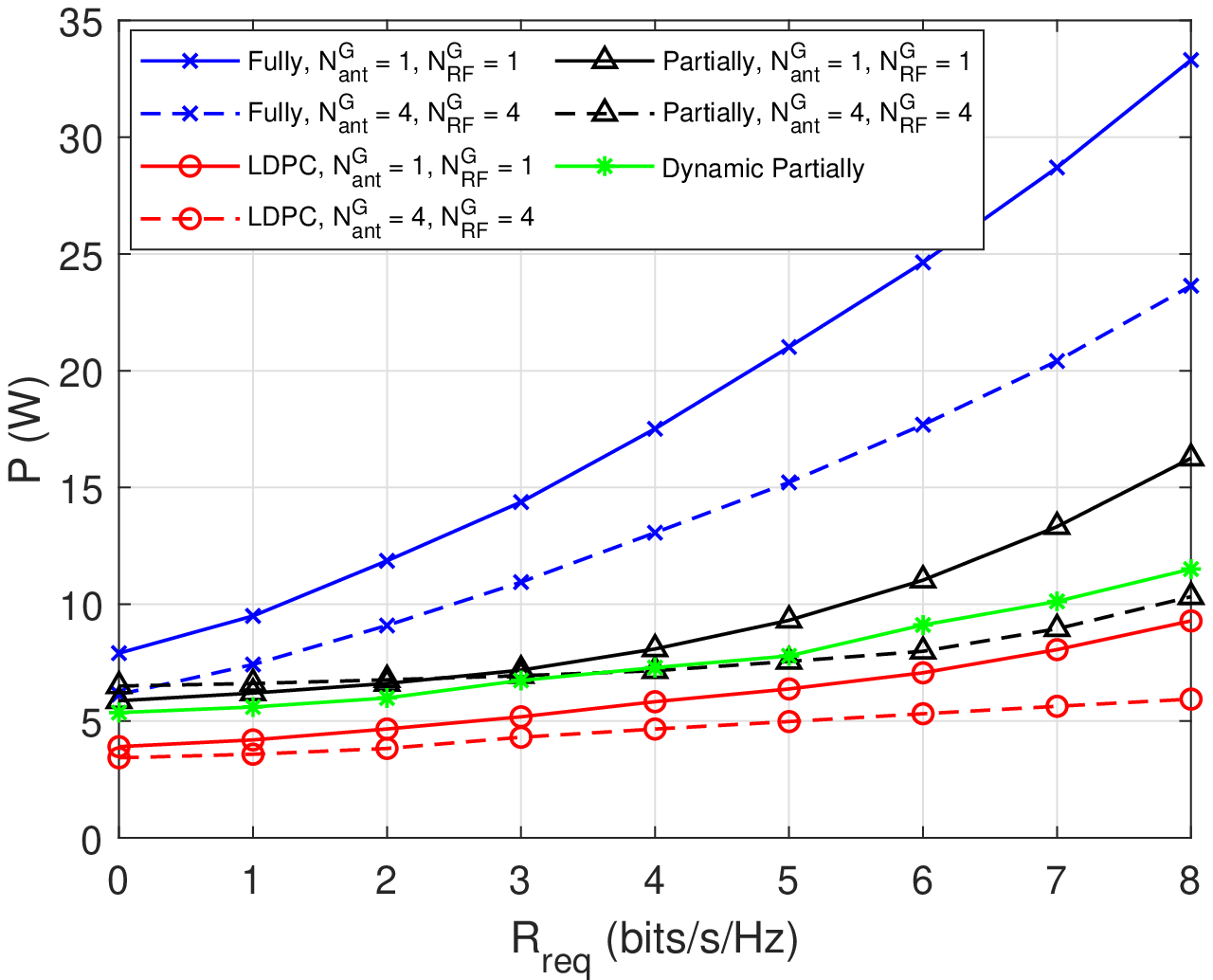} \label{qos1}}	
\subfigure[]{\includegraphics[width=\m in]{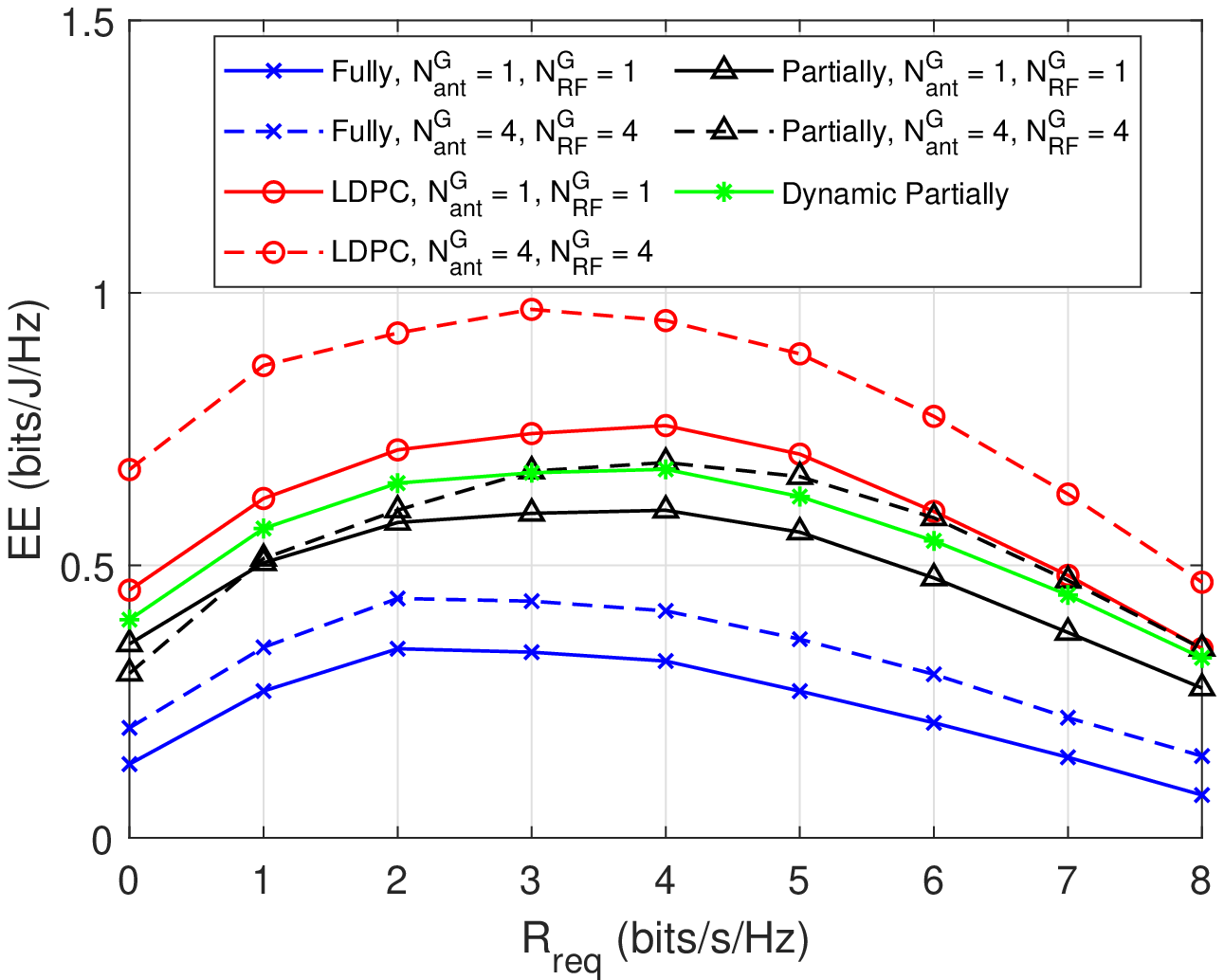} \label{qos3}}
\caption{Performance comparison under different architectures of fully, partially, dynamic partially and LDPC-based connections with $N_{ant}^{G}=N_{RF}^{G}\in\{1,4\}$ RF/antenna elements per group and QoS requirements of $R_{req}\in\{0,1,2,3,4,5,6,7,8\}$ bits/s/Hz in terms of (a) power consumption and (b) EE. The dynamic partial connections are adjustable in each subarray size; hence, the parameters of $N_{ant}^{G}=N_{RF}^{G}$ are not set.} \label{qos}
\end{figure}

In Fig. \ref{qos}, we compare the performance under different architectures, i.e., full connection \cite{29,30,31}, fixed partial connection \cite{24,25,26}, the dynamic partial connection designed in \cite{new1} and the proposed LDPC-based architecture considering different numbers of RF/antenna elements per group and QoS requirements. We also compare two different types of partial connections. Dynamic partial connection in \cite{new1} is performed by exhaustively obtaining the optimal subarray patterns, i.e., each subarray has different sizes of elements but is fully connected in its group with the same connection as in the fixed case. This potentially provides a higher degree of freedom to construct a high-EE partial connection, which accordingly achieves a higher EE with similar power consumption to the fixed partially connected architecture. Moreover, as explained previously, it requires more power to satisfy higher QoS demands, especially under the fully connected architecture with excessive power consumption up to approximately $33.3$ W. Partial connection results in approximately $6.9$ W more power consumption than the proposed LDPC-based architecture, since the controllers with partial connections cannot acquire certain messages passed from some nodes. Furthermore, a larger difference in power is induced by higher QoS requirements by comparing the distributed control architecture with $N_{ant}^G = N_{RF}^G = 1$ and a centralized architecture with $N_{ant}^G=N_{RF}^G=4$. This is because some information in the distributed architecture has to be conveyed from faraway nodes with massive hops, generating an out-of-date policy to fulfill high QoS requirements, as also shown in Fig. \ref{ng}. To elaborate a little further, we observe opposite trends for the partially connected method. Under lower QoS $R_{req}\leq 4$ bits/s/Hz, controllers in a more distributed arrangement, i.e., $N_{ant}^G = N_{RF}^G = 1$, tend to greedily optimize themselves due to unknown or nonupdated messages passed from others. In contrast, a more centralized arrangement with $N_{ant}^G = N_{RF}^G = 4$ is able to provide a more appropriate policy with adequate QoS-aware information, achieving lower power and, accordingly, higher EE. In conclusion, the proposed LDPC-based architecture under the MARS scheme can conserve approximately $24$ W and $6.9$ W and achieves approximately $2.2$ and $1.4$ times EE compared to fully and partially connected structures, respectively.

\begin{figure*}
\centering
\subfigure[]{\includegraphics[width=\m in]{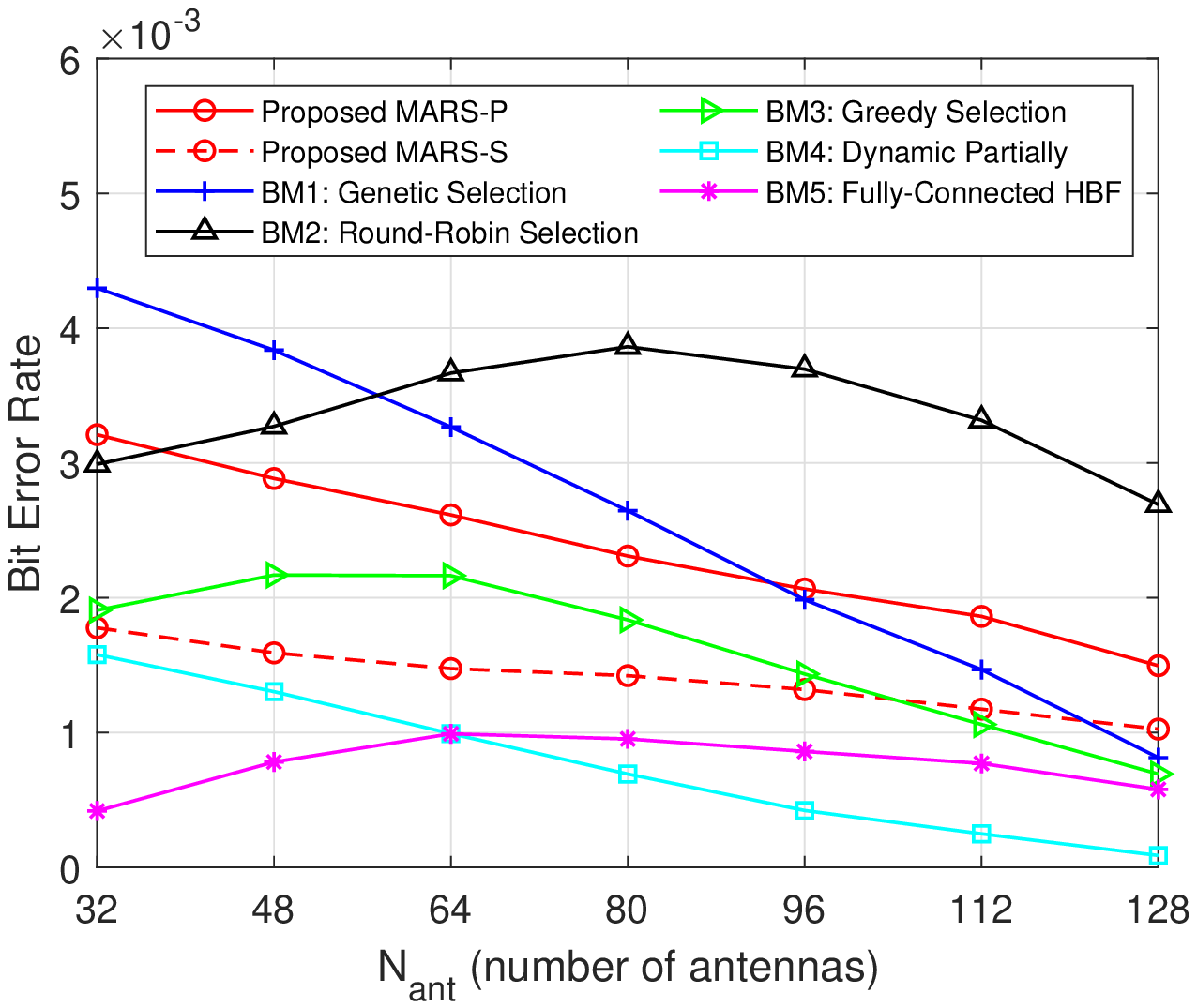} \label{bmber}}
\subfigure[]{\includegraphics[width=\m in]{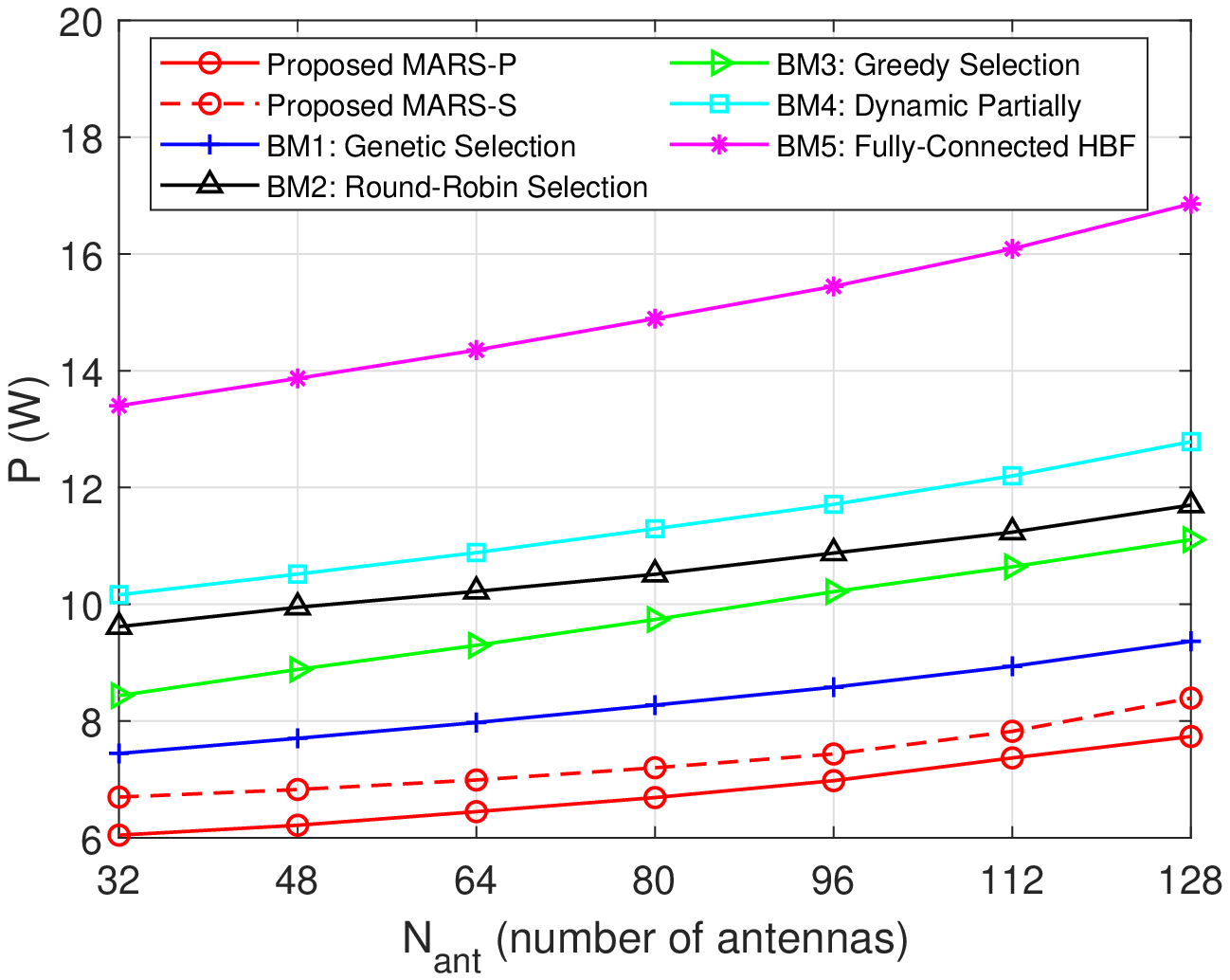} \label{bm1}}
\subfigure[]{\includegraphics[width=\m in]{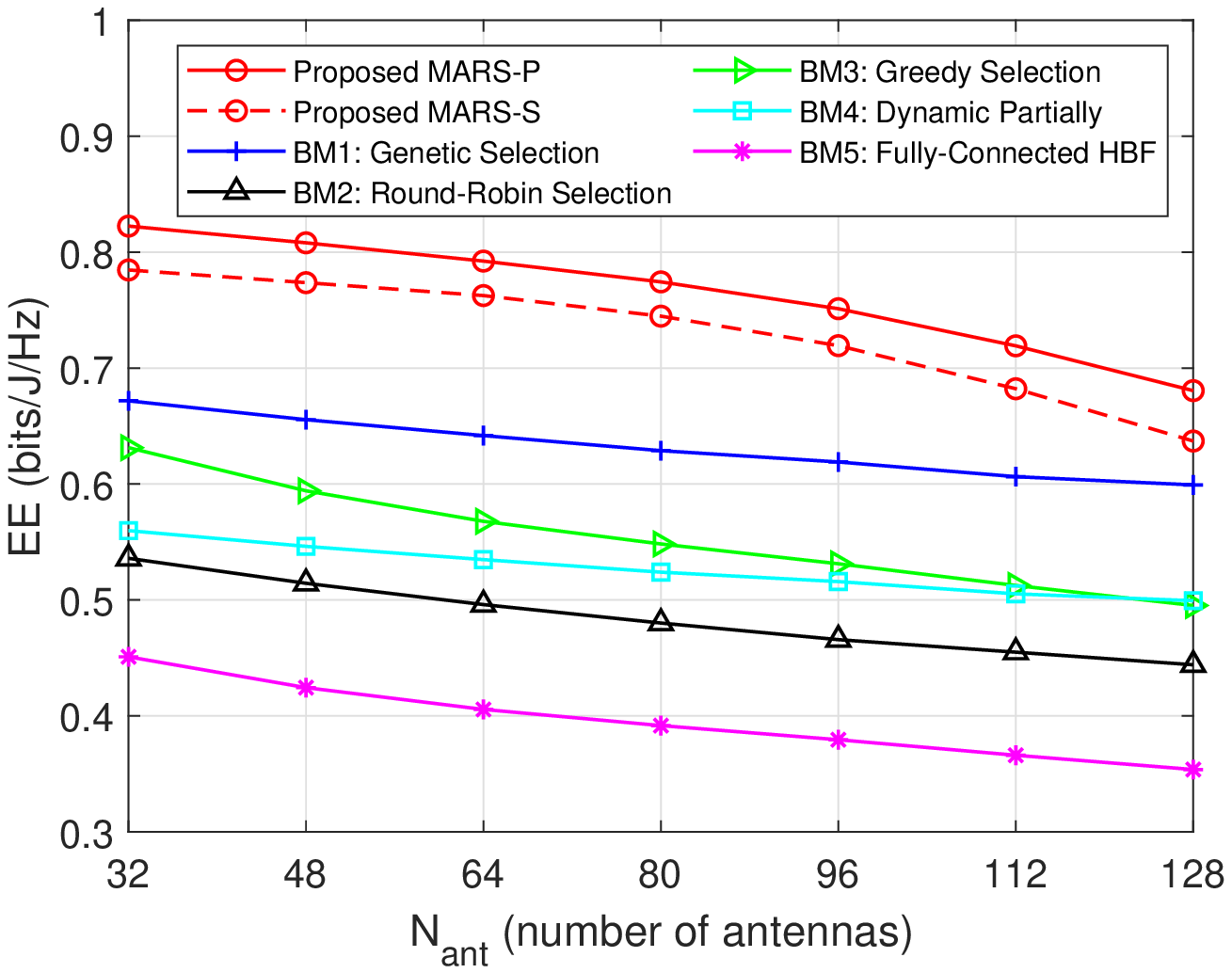} \label{bm3}}
\subfigure[]{\includegraphics[width=\m in]{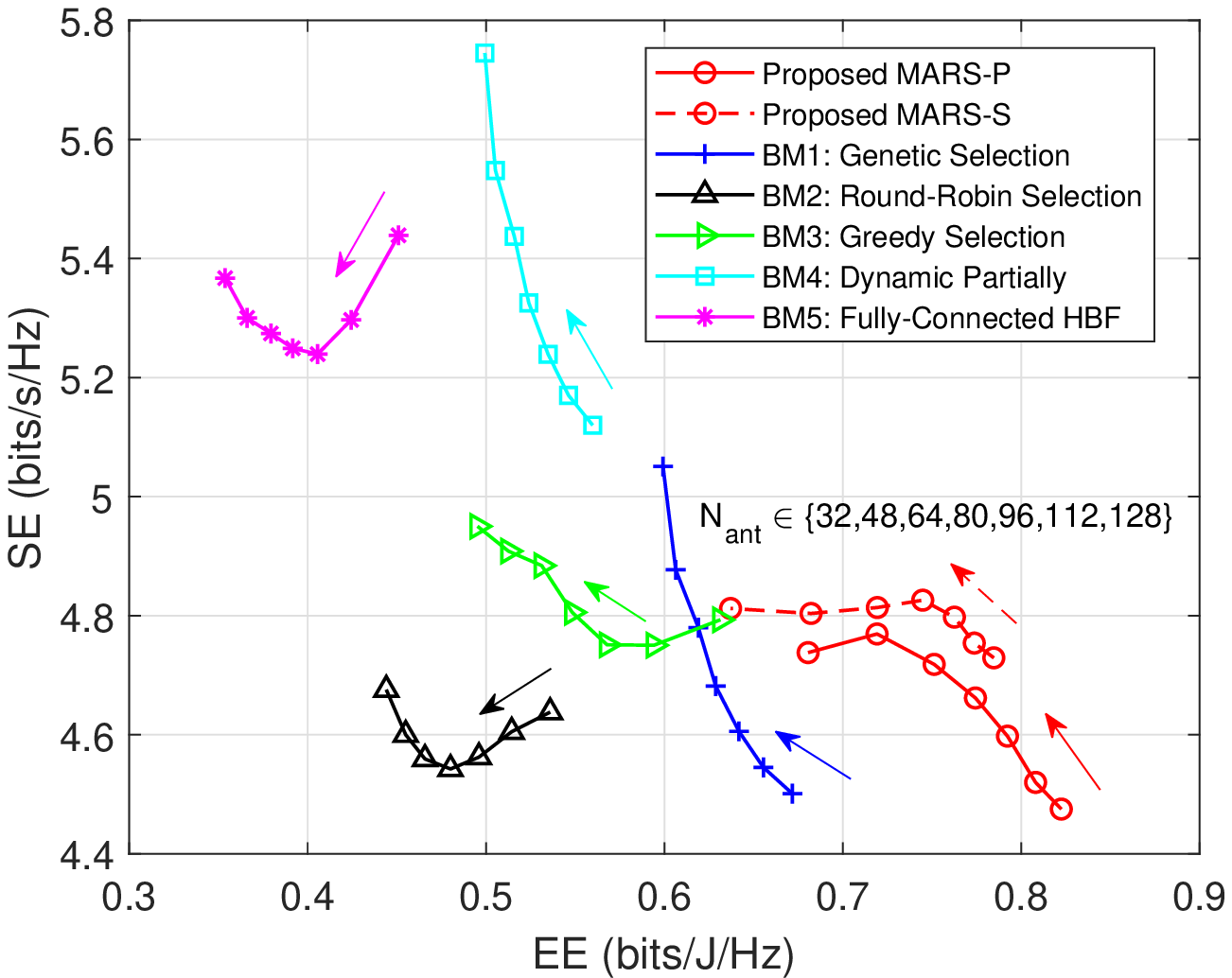} \label{bmsee}}
\caption{Performance comparison of the proposed MARS scheme and benchmarks of genetic, greedy and round-robin selection with different numbers of receiving antennas $N_{ant}\in\{32,48,64,80,96,112,128\}$ in terms of (a) bit error rate, (b) power consumption, (c) EE, and (d) spectrum-energy efficiency. The direction of the arrows indicates the increment of $N_{ant}$.} \label{bm}
\end{figure*}

	In Fig. \ref{bm}, we compare the comprehensive performance in terms of the bit error rate (BER), power consumption, EE, and spectrum-energy efficiency for the proposed MARS schemes MARS-P and MARS-S with five benchmarks (BMs), with are elaborated as follows. Note that modulation of 64-quadrature modulation (64-QAM) is adopted for computing BER. We consider an equivalent QoS constraint $R_{req}=3$ bits/s/Hz, insertion loss $\beta=0.3$, $N_{RF}=32$, and $N_{RF}^G=N_{ant}^G=4$ for a fair comparison.
\begin{itemize}
	\item \textbf{BM1 Genetic-Based Selection} \cite{50}: The selection solution is initially provided with a given genetic population size. The candidate solution is obtained by adopting a discrete genetic algorithm with crossover, mutation, elite selection and offspring gene generation. All RF chains are selected to operate, i.e., $\delta_n=1\ \forall n\in \mathcal{R}$. This case is conducted based on the designed LDPC-based connection.
	
	\item \textbf{BM2 Round-Robin Selection}: The candidate solution of RF/antenna selection is randomly initialized and evaluated in a round-robin manner. We change the mechanism in \cite{bm101} from MIMO user scheduling to RF/antenna selection. The policy under fully off antennas or RF chains is excluded here due to zero rate performance. This case is conducted based on the designed LDPC-based connection.
	
	\item \textbf{BM3 Greedy-Based Selection} \cite{49}: It optimizes the policy of a single antenna with the other solutions fixed, which can be regarded as a distributed approach with $N_{ant}^G=1$ in this work. First, we randomly conduct antenna selection by either turning on or off all nodes. Afterwards, we iteratively select the antenna that leads to the lowest power consumption in a selfish manner until the completion of the final node. All RF chains are selected to be under operation, i.e., $\delta_n=1\ \forall n\in \mathcal{R}$. This case is conducted based on the designed LDPC-based connection.
	
	\item \textbf{BM4 Dynamic Partial Connection} \cite{new1}: This connection is implemented by exhaustively obtaining the optimal subarray patterns, i.e., each subarray has elements of different sizes but is fully connected in its subgroup as in the fixed case, which provides a higher degree of freedom than partial connection with an equal subarray size. All RF chains and antennas in each subarray are selected to be operated.
	
	\item \textbf{BM5 Fully Connected HBF} \cite{bm100}: The conventional HBF architecture is implemented under a full connection with whole RF chains and antennas selected to be operated.
	
\end{itemize}
In Fig. \ref{bmber}, we observe that BMs 4 and 5 have the lowest BERs due to more connected links in HBF, i.e., more information with a higher degree of freedom can be obtained to achieve a higher rate. BMs 1 and 3 show medium BER performance under the LDPC architecture. With similar performance, the proposed MARS scheme has a lower BER with more antennas due to more information being passed. Round robin potentially switches off beneficial nodes, leading to the worst BER. To elaborate a little further, BER reflects the negative correlation in the performance in terms of spectrum efficiency (SE), as shown in Fig. \ref{bmsee}, i.e., a lower BER corresponds to a higher SE. We infer from Fig. \ref{bm1} that more power is required to operate the increased number of deployed antennas. Considering the LDPC architecture, however, BM2 requires the highest power due to random selection without any information utilization. For BM3, a selfish-style method is employed to optimize its own antenna selection without using messages from other controller groups. However, as a compromise mechanism in BM1, genetic-based selection takes into account the selection policy from other groups, which preserves powers of $2.24$ W and $1.74$ W compared to round-robin and greedy selection, respectively, under $N_{ant}=128$. In contrast, BM4, regarded as the optimal solution with partial connection, requires slightly more power than the baseline solution in the LDPC architecture due to comparably fewer links being required to achieve satisfactory service. Moreover, BM5 with full connection consumes the most power resources, approximately twice as much as the proposed MARS schemes.

Although MARS-P has a faster policy update speed, it potentially leads to the missed information due to the parallel transfer in message passing, which has slightly higher power consumption of approximately $0.66$ W compared to that of the sequential approach in MARS-S. Benefiting from both the LDPC connection and message passing-based design, the proposed MARS can exactly convey appropriate information to all the other controller groups to realize a better RF/antenna selection policy. In this context, as observed from Figs. \ref{bm1} and \ref{bm3} considering $N_{ant}=128$, the proposed MARS scheme outperforms all the benchmarks, which is able to preserve powers of approximately $1.6$, $4$, $3.4$, $5.1$ and $9.1$ W and to improve EE by approximately $13\%$, $53\%$, $37\%$, $36\%$ and $93\%$ compared to BMs 1 to 5, respectively. Furthermore, as shown in Fig. \ref{bmsee}, it also strikes a tradeoff between the SE and EE metrics. MARS accomplishes the highest EE at the expense of moderate throughput, while full-connection and dynamic partial-connection have the higher rate by consuming more power, leading to the compellingly lower EE. Additionally, BMs 1 and 4 exhibit asymptotic shaped curves because the optimal solution is potentially obtained from genetic selection with a smaller search space.

\begin{figure}
\centering
\includegraphics[width=3.3in]{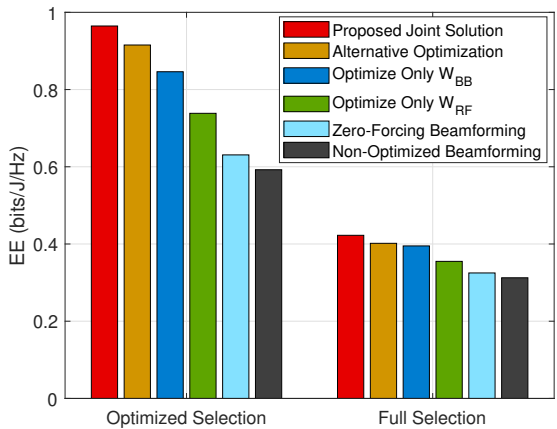}
\caption{Comparison of different beamforming methods for $\mathbf{W_{BB}}$ and $\mathbf{W_{RF}}$ with $N_{ant}=64$ and $N_{RF}=32$.} \label{bf}
\end{figure}

In Fig. \ref{bf}, we compare different HBF methods, including genetic-based algorithm, alternative optimization, and zero-forcing for providing appropriate beamformers of $\mathbf{W_{BB}}$ and $\mathbf{W_{RF}}$ with or without the optimized RF/antenna selection mechanism. The MARS-P is performed for the optimized selection, whilst full selection turns on all antennas and RF chains. Genetic-based algorithm provides the joint solution by considering both beamformers at the same time through the process of gene generation, elite selection crossover and mutation. In alternative optimization, the two parameters undergo rotational optimization, i.e., $\mathbf{W_{BB}}$ is optimized based on the optimal $\mathbf{W_{RF}}$ from previous iteration. Note that optimizing $\mathbf{W_{BB}}$ indicates the method only optimizes $\mathbf{W_{BB}}$ with random $\mathbf{W_{RF}}$. Same process takes place for $\mathbf{W_{RF}}$-only optimization. Random method is conducted for the non-optimized beamforming. In Fig. \ref{bf}, we can readily observe that it has a lower EE due to high power consumption of full selection. Performance gain of additional beamforming can also be achieved, i.e., joint optimization has improved EE with approximately $5.3\%$, $14\%$, $30.6\%$, $52.9\%$, $62.8\%$ compared to alternative optimization, optimizing $\mathbf{W_{BB}}$/$\mathbf{W_{RF}}$-only, zero-forcing, non-optimized beamforming, respectively. This is because more power can be preserved upon satisfaction of rate requirement.

\section{Conclusions}\label{ch:Conclusion}
	In this paper, we propose an LDPC-based HBF structure and MARS scheme to jointly consider HBF as well as RF chain and antenna selection for the minimization of operating power consumption, which guarantees QoS and available power utilization. MARS can be employed under sequential (MARS-S) and parallel (MARS-P) message passing under either a distributed or centralized architecture. Moreover, heuristic HBF scheme is designed based on the continuous genetic algorithm. The performance reveals that MARS-P has faster convergence than MARS-S due to parallel information delivery. However, MARS-S achieves lower power consumption and higher EE due to sequential decisions from fewer neighboring nodes. Additionally, out-of-date messages of policies are solved by selecting appropriate learning rates for RF chains and antenna selections. The overall performance is more significantly influenced by RF chains than antennas under different numbers of controllers, connections, and hardware impairment effects. Moreover, the proposed LDPC-based structure with MARS achieves the lowest power and highest EE since it possesses considerably fewer links than the fully connected architecture but with more information exchanged than the partially connected architecture. Additionally, benefiting from better messages and policies passed from other RF/antenna nodes, MARS outperforms the existing algorithms, namely, round-robin, greedy-based, and genetic-based selection methods, as well as dynamic adjustment of partial/full connection, in the open literature. Moreover, the joint solution leveraging both HBF and RF/antenna selection schemes accomplishes the highest EE performance.

\bibliographystyle{IEEEtran}
\bibliography{IEEEabrv}

\begin{thebibliography}{10}
\providecommand{\url}[1]{#1}
\csname url@samestyle\endcsname
\providecommand{\newblock}{\relax}
\providecommand{\bibinfo}[2]{#2}
\providecommand{\BIBentrySTDinterwordspacing}{\spaceskip=0pt\relax}
\providecommand{\BIBentryALTinterwordstretchfactor}{4}
\providecommand{\BIBentryALTinterwordspacing}{\spaceskip=\fontdimen2\font plus
\BIBentryALTinterwordstretchfactor\fontdimen3\font minus
  \fontdimen4\font\relax}
\providecommand{\BIBforeignlanguage}[2]{{%
\expandafter\ifx\csname l@#1\endcsname\relax
\typeout{** WARNING: IEEEtran.bst: No hyphenation pattern has been}%
\typeout{** loaded for the language `#1'. Using the pattern for}%
\typeout{** the default language instead.}%
\else
\language=\csname l@#1\endcsname
\fi
#2}}
\providecommand{\BIBdecl}{\relax}
\BIBdecl

\bibitem{1}
W.~Saad, M.~Bennis, and M.~Chen, ``A vision of {6G} wireless systems:
  Applications, trends, technologies, and open research problems,'' \emph{IEEE
  Network}, vol.~34, no.~3, pp. 134--142, 2020.

\bibitem{acm}
L.-H. Shen, K.-T. Feng, and L.~Hanzo, ``Five facets of {6G}: Research
  challenges and opportunities,'' \emph{ACM Computing Surveys}, vol.~55,
  no.~11, pp. 1--39, 2023.

\bibitem{2}
H.~Zhang, S.~Huang, C.~Jiang, K.~Long, V.~C.~M. Leung, and H.~V. Poor, ``Energy
  efficient user association and power allocation in millimeter-wave-based
  ultra dense networks with energy harvesting base stations,'' \emph{IEEE
  Journal on Selected Areas in Communications}, vol.~35, no.~9, pp. 1936--1947,
  2017.

\bibitem{3}
T.~Huang, W.~Yang, J.~Wu, J.~Ma, X.~Zhang, and D.~Zhang, ``A survey on green
  {6G} network: Architecture and technologies,'' \emph{IEEE Access}, vol.~7,
  pp. 175\,758--175\,768, 2019.

\bibitem{4}
S.~Yang and L.~Hanzo, ``Fifty years of {MIMO} detection: The road to
  large-scale {MIMOs},'' \emph{IEEE Communications Surveys Tutorials}, vol.~17,
  no.~4, pp. 1941--1988, 2015.

\bibitem{5}
Q.-U.-A. Nadeem, A.~Kammoun, and M.-S. Alouini, ``Elevation beamforming with
  full dimension {MIMO} architectures in {5G} systems: A tutorial,'' \emph{IEEE
  Communications Surveys Tutorials}, vol.~21, no.~4, pp. 3238--3273, 2019.

\bibitem{7}
S.~Kutty and D.~Sen, ``Beamforming for millimeter wave communications: An
  inclusive survey,'' \emph{IEEE Communications Surveys Tutorials}, vol.~18,
  no.~2, pp. 949--973, 2016.

\bibitem{9}
S.~Han, C.-l. I, Z.~Xu, and C.~Rowell, ``Large-scale antenna systems with
  hybrid analog and digital beamforming for millimeter wave {5G},'' \emph{IEEE
  Communications Magazine}, vol.~53, no.~1, pp. 186--194, 2015.

\bibitem{10}
S.~S. Ioushua and Y.~C. Eldar, ``A family of hybrid analog–digital
  beamforming methods for massive {MIMO} systems,'' \emph{IEEE Transactions on
  Signal Processing}, vol.~67, no.~12, pp. 3243--3257, 2019.

\bibitem{11}
I.~Ahmed, H.~Khammari, A.~Shahid, A.~Musa, K.~S. Kim, E.~De~Poorter, and
  I.~Moerman, ``A survey on hybrid beamforming techniques in {5G}: Architecture
  and system model perspectives,'' \emph{IEEE Communications Surveys
  Tutorials}, vol.~20, no.~4, pp. 3060--3097, 2018.

\bibitem{14}
L.-H. Shen and K.-T. Feng, ``Mobility-aware subband and beam resource
  allocation schemes for millimeter wave wireless networks,'' \emph{IEEE
  Transactions on Vehicular Technology}, vol.~69, no.~10, pp. 11\,893--11\,908,
  2020.

\bibitem{15}
Y.~Cai, Y.~Xu, Q.~Shi, B.~Champagne, and L.~Hanzo, ``Robust joint hybrid
  transceiver design for millimeter wave full-duplex {MIMO} relay systems,''
  \emph{IEEE Transactions on Wireless Communications}, vol.~18, no.~2, pp.
  1199--1215, 2019.

\bibitem{17}
L.-H. Shen, Y.-C. Chen, and K.-T. Feng, ``Design and analysis of multi-user
  association and beam training schemes for millimeter wave based {WLANs},''
  \emph{IEEE Transactions on Vehicular Technology}, vol.~69, no.~7, pp.
  7458--7472, 2020.

\bibitem{18}
L.-H. Shen, K.-T. Feng, and L.~Hanzo, ``Coordinated multiple access point
  multiuser beamforming training protocol for millimeter wave {WLANs},''
  \emph{IEEE Transactions on Vehicular Technology}, vol.~69, no.~11, pp.
  13\,875--13\,889, 2020.

\bibitem{21}
L.-H. Shen, T.-W. Chang, K.-T. Feng, and P.-T. Huang, ``Design and
  implementation for deep learning based adjustable beamforming training for
  millimeter wave communication systems,'' \emph{IEEE Transactions on Vehicular
  Technology}, vol.~70, no.~3, pp. 2413--2427, 2021.

\bibitem{22}
H.~Li, M.~Li, Q.~Liu, and A.~L. Swindlehurst, ``Dynamic hybrid beamforming with
  low-resolution {PSs} for wideband mmwave {MIMO-OFDM} systems,'' \emph{IEEE
  Journal on Selected Areas in Communications}, vol.~38, no.~9, pp. 2168--2181,
  2020.

\bibitem{23}
Y.~Zhang, J.~Du, Y.~Chen, M.~Han, and X.~Li, ``Optimal hybrid beamforming
  design for millimeter-wave massive multi-user {MIMO} relay systems,''
  \emph{IEEE Access}, vol.~7, pp. 157\,212--157\,225, 2019.

\bibitem{24}
Y.~Zhang, J.~Du, Y.~Chen, X.~Li, K.~M. Rabie, and R.~Khkrel, ``Dual-iterative
  hybrid beamforming design for millimeter-wave massive multi-user {MIMO}
  systems with sub-connected structure,'' \emph{IEEE Transactions on Vehicular
  Technology}, vol.~69, no.~11, pp. 13\,482--13\,496, 2020.

\bibitem{newb1}
H.~Gao, D.~Liu, and Z.~Zhang, ``Low-resolution {DACs} aided partially-connected
  hybrid precoding design for {MU-MISO} system,'' \emph{IEEE Communications
  Letters}, pp. 1--1, 2023.

\bibitem{newb2}
X.~Wang, Z.~Fei, J.~A. Zhang, and J.~Xu, ``Partially-connected hybrid
  beamforming design for integrated sensing and communication systems,''
  \emph{IEEE Transactions on Communications}, vol.~70, no.~10, pp. 6648--6660,
  2022.

\bibitem{28}
A.~Morsali, A.~Haghighat, and B.~Champagne, ``Generalized framework for hybrid
  analog/digital signal processing in massive and ultra-massive-{MIMO}
  systems,'' \emph{IEEE Access}, vol.~8, pp. 100\,262--100\,279, 2020.

\bibitem{27}
D.~Zhang, Y.~Wang, X.~Li, and W.~Xiang, ``Hybridly connected structure for
  hybrid beamforming in mmwave massive {MIMO} systems,'' \emph{IEEE
  Transactions on Communications}, vol.~66, no.~2, pp. 662--674, 2018.

\bibitem{new1}
Y.~Chen, D.~Chen, T.~Jiang, and L.~Hanzo, ``Millimeter-wave massive {MIMO}
  systems relying on generalized sub-array-connected hybrid precoding,''
  \emph{IEEE Transactions on Vehicular Technology}, vol.~68, no.~9, pp.
  8940--8950, 2019.

\bibitem{29}
N.~N. Moghadam, G.~Fodor, M.~Bengtsson, and D.~J. Love, ``On the energy
  efficiency of {MIMO} hybrid beamforming for millimeter-wave systems with
  nonlinear power amplifiers,'' \emph{IEEE Transactions on Wireless
  Communications}, vol.~17, no.~11, pp. 7208--7221, 2018.

\bibitem{30}
M.~Sefunç, A.~Zappone, and E.~A. Jorswieck, ``Energy efficiency of mmwave
  {MIMO} systems with spatial modulation and hybrid beamforming,'' \emph{IEEE
  Transactions on Green Communications and Networking}, vol.~4, no.~1, pp.
  95--108, 2020.

\bibitem{31}
L.~Zhao, M.~Li, C.~Liu, S.~V. Hanly, I.~B. Collings, and P.~A. Whiting,
  ``Energy efficient hybrid beamforming for multi-user millimeter wave
  communication with low-resolution {A/D} at transceivers,'' \emph{IEEE Journal
  on Selected Areas in Communications}, vol.~38, no.~9, pp. 2142--2155, 2020.

\bibitem{32}
A.~Kaushik, E.~Vlachos, C.~Tsinos, J.~Thompson, and S.~Chatzinotas, ``Joint bit
  allocation and hybrid beamforming optimization for energy efficient
  millimeter wave {MIMO} systems,'' \emph{IEEE Transactions on Green
  Communications and Networking}, vol.~5, no.~1, pp. 119--132, 2021.

\bibitem{newisac1}
C.~Qi, W.~Ci, J.~Zhang, and X.~You, ``Hybrid beamforming for millimeter wave
  {MIMO} integrated sensing and communications,'' \emph{IEEE Communications
  Letters}, vol.~26, no.~5, pp. 1136--1140, 2022.

\bibitem{newisac2}
X.~Yu, L.~Tu, Q.~Yang, M.~Yu, Z.~Xiao, and Y.~Zhu, ``Hybrid beamforming in
  mmwave massive {MIMO} for {IoV} with dual-functional radar communication,''
  \emph{IEEE Transactions on Vehicular Technology}, vol.~72, no.~7, pp.
  9017--9030, 2023.

\bibitem{newnew1}
L.~Xu, S.~Sun, Y.~D. Zhang, and A.~Petropulu, ``Joint antenna selection and
  beamforming in integrated automotive radar sensing-communications with
  quantized double phase shifters,'' in \emph{Proc. IEEE International
  Conference on Acoustics, Speech and Signal Processing (ICASSP)}, 2023, pp.
  1--5.

\bibitem{newnew2}
A.~M. Elbir, A.~Abdallah, A.~Celik, and A.~M. Eltawil, ``Antenna selection with
  beam squint compensation for integrated sensing and communications,''
  \emph{eprint arXiv:2307.07242}, 2023.

\bibitem{33}
S.~Payami, M.~Ghoraishi, and M.~Dianati, ``Hybrid beamforming for large antenna
  arrays with phase shifter selection,'' \emph{IEEE Transactions on Wireless
  Communications}, vol.~15, no.~11, pp. 7258--7271, 2016.

\bibitem{36}
A.~Kaushik, J.~Thompson, E.~Vlachos, C.~Tsinos, and S.~Chatzinotas, ``Dynamic
  {RF} chain selection for energy efficient and low complexity hybrid
  beamforming in millimeter wave {MIMO} systems,'' \emph{IEEE Transactions on
  Green Communications and Networking}, vol.~3, no.~4, pp. 886--900, 2019.

\bibitem{37}
E.~Vlachos and J.~Thompson, ``Energy-efficiency maximization of hybrid massive
  {MIMO} precoding with random-resolution {DACs} via {RF} selection,''
  \emph{IEEE Transactions on Wireless Communications}, vol.~20, no.~2, pp.
  1093--1104, 2021.

\bibitem{38}
F.~Kschischang, B.~Frey, and H.-A. Loeliger, ``Factor graphs and the
  sum-product algorithm,'' \emph{IEEE Transactions on Information Theory},
  vol.~47, no.~2, pp. 498--519, 2001.

\bibitem{39}
J.~Zeng, J.~Lin, and Z.~Wang, ``Low complexity message passing detection
  algorithm for large-scale {MIMO} systems,'' \emph{IEEE Wireless
  Communications Letters}, vol.~7, no.~5, pp. 708--711, 2018.

\bibitem{yang1}
Z.~Sui, S.~Yan, H.~Zhang, L.-L. Yang, and L.~Hanzo, ``Approximate message
  passing algorithms for low complexity {OFDM-IM} detection,'' \emph{IEEE
  Transactions on Vehicular Technology}, vol.~70, no.~9, pp. 9607--9612, 2021.

\bibitem{42}
N.~J. Myers, J.~Kaleva, A.~Tölli, and R.~W. Heath, ``Message passing-based
  link configuration in short range millimeter wave systems,'' \emph{IEEE
  Transactions on Communications}, vol.~68, no.~6, pp. 3465--3479, 2020.

\bibitem{43}
Y.~Jeong, C.~Lee, and Y.~H. Kim, ``Power minimizing beamforming and power
  allocation for {MISO-NOMA} systems,'' \emph{IEEE Transactions on Vehicular
  Technology}, vol.~68, no.~6, pp. 6187--6191, 2019.

\bibitem{44}
N.-I. Kim and D.-H. Cho, ``Scheduling and layer selection based performance
  improvement of uplink {SCMA} system with receive beamforming,'' \emph{IEEE
  Transactions on Vehicular Technology}, vol.~68, no.~7, pp. 7209--7213, 2019.

\bibitem{45}
Y.~Yu, M.~Pischella, and D.~Le~Ruyet, ``Distributed antenna selection with
  message passing algorithm for {MIMO D2D} communications,'' in \emph{Proc.
  IEEE International Symposium on Wireless Communication Systems (ISWCS)},
  2017, pp. 454--458.

\bibitem{46}
R.~Gallager, ``Low-density parity-check codes,'' \emph{IRE Transactions on
  Information Theory}, vol.~8, no.~1, pp. 21--28, 1962.

\bibitem{47}
T.~Richardson, M.~Shokrollahi, and R.~Urbanke, ``Design of capacity-approaching
  irregular low-density parity-check codes,'' \emph{IEEE Transactions on
  Information Theory}, vol.~47, no.~2, pp. 619--637, 2001.

\bibitem{48}
O.~Hiari and R.~Mesleh, ``Impact of {RF}–switch insertion loss on the
  performance of space modulation techniques,'' \emph{IEEE Communications
  Letters}, vol.~22, no.~5, pp. 958--961, 2018.

\bibitem{cga}
R.~Chelouah and P.~Siarry, ``A continuous genetic algorithm designed for the
  global optimization of multimodal functions,'' \emph{Journal of Heuristics},
  vol.~6, no.~2, pp. 191--213, Jun. 2000.

\bibitem{49}
M.~O.~K. Mendonça, P.~S.~R. Diniz, T.~N. Ferreira, and L.~Lovisolo, ``Antenna
  selection in massive {MIMO} based on greedy algorithms,'' \emph{IEEE
  Transactions on Wireless Communications}, vol.~19, no.~3, pp. 1868--1881,
  2020.

\bibitem{50}
J.~C. Marinello, T.~Abrão, A.~Amiri, E.~de~Carvalho, and P.~Popovski,
  ``Antenna selection for improving energy efficiency in {XL-MIMO} systems,''
  \emph{IEEE Transactions on Vehicular Technology}, vol.~69, no.~11, pp.
  13\,305--13\,318, 2020.

\bibitem{uepower}
J.~Brun, V.~Palhares, G.~Marti, and C.~Studer, ``Beam alignment for the
  cell-free mmwave massive {MU-MIMO} uplink,'' in \emph{Proc. IEEE Workshop on
  Signal Processing Systems (SiPS)}, 2022, pp. 1--6.

\bibitem{12}
R.~Méndez-Rial, C.~Rusu, N.~González-Prelcic, A.~Alkhateeb, and R.~W. Heath,
  ``Hybrid {MIMO} architectures for millimeter wave communications: Phase
  shifters or switches?'' \emph{IEEE Access}, vol.~4, pp. 247--267, 2016.

\bibitem{bb}
H.~Krishna, R.~Rajakumar, and S.~Chakrabarti, ``Quantifying the improvement in
  energy savings for {LTE} enodeb baseband subsystem with technology scaling
  and multi-core architectures,'' in \emph{Proc. National Conference on
  Communications (NCC)}, 2012, pp. 1--5.

\bibitem{lna}
M.~Kadam, S.~Aniruddhan, and A.~Kumar, ``An unconditionally stable 28 ghz 18 db
  gain {LNA} employing current-reuse,'' in \emph{Proc. IEEE International
  Symposium on Circuits and Systems (ISCAS)}, 2020, pp. 1--4.

\bibitem{pslna}
P.~Skrimponis, S.~Dutta, M.~Mezzavilla, S.~Rangan, S.~H. Mirfarshbafan,
  C.~Studer, J.~Buckwalter, and M.~Rodwell, ``Power consumption analysis for
  mobile mmwave and sub-{THz} receivers,'' in \emph{Proc. IEEE 6G Wireless
  Summit (6G SUMMIT)}, 2020, pp. 1--5.

\bibitem{spec}
3GPP, ``Study on channel model for frequencies from 0.5 to 100 {GHz},'' no. TR
  38.901 version 16.1.0 Release 16, Nov. 2020.

\bibitem{25}
X.~Song, T.~Kühne, and G.~Caire, ``Fully-/partially-connected hybrid
  beamforming architectures for mmwave {MU-MIMO},'' \emph{IEEE Transactions on
  Wireless Communications}, vol.~19, no.~3, pp. 1754--1769, 2020.

\bibitem{26}
N.~T. Nguyen and K.~Lee, ``Unequally sub-connected architecture for hybrid
  beamforming in massive {MIMO} systems,'' \emph{IEEE Transactions on Wireless
  Communications}, vol.~19, no.~2, pp. 1127--1140, 2020.

\bibitem{bm101}
H.~A. Ammar, R.~Adve, S.~Shahbazpanahi, G.~Boudreau, and K.~V. Srinivas,
  ``Downlink resource allocation in multiuser cell-free {MIMO} networks with
  user-centric clustering,'' \emph{IEEE Transactions on Wireless
  Communications}, vol.~21, no.~3, pp. 1482--1497, 2022.

\bibitem{bm100}
B.-Y. Chen, Y.-F. Chen, and S.-M. Tseng, ``Hybrid beamforming and data stream
  allocation algorithms for power minimization in multi-user massive
  {MIMO-OFDM} systems,'' \emph{IEEE Access}, vol.~10, pp. 101\,898--101\,912,
  2022.

\end{thebibliography}

\end{document}